\UseRawInputEncoding
\documentclass[11pt,reqno]{amsart}
\allowdisplaybreaks[4]
\usepackage{graphicx}
\usepackage{subfigure}
\usepackage{epsfig}
\usepackage{amssymb}
\usepackage{mathtools}
\usepackage{amsmath,xypic}
\usepackage{cite,color}
\usepackage{cite}
\usepackage{amsmath}
\usepackage{tikz}
\usepackage{epic,eepic,amsmath,amsthm,amssymb}
\usepackage{epsfig,cite,indentfirst,oldgerm}
\usepackage{multicol,boxedminipage}
\usepackage{epic,eepic,amsmath,amsthm,amssymb}
\usepackage{epsfig,cite,indentfirst,oldgerm}
\usepackage{multicol,boxedminipage}
\usepackage{array,caption}
\usepackage{enumerate}
\usepackage{booktabs}
\numberwithin{equation}{section}

\newtheorem{theorem}{Theorem}[section]
\newtheorem{definition}[theorem]{Definition}
\newtheorem{proposition}[theorem]{Proposition}

\newtheorem{lemma}[theorem]{Lemma}
\newtheorem{remark}[theorem]{Remark}

\newcommand{\be}{\begin{equation}}
\newcommand{\ee}{\end{equation}}
\newcommand{\bea}{\begin{eqnarray}}
\newcommand{\eea}{\end{eqnarray}}
\newcommand{\ba}{\begin{array}}
	\newcommand{\ea}{\end{array}}
\newcommand{\bean}{\begin{eqnarray*}}
	\newcommand{\eean}{\end{eqnarray*}}

	\def\lprod{\mathop{\prod{\mkern-29.5mu}{\mathbf\longleftarrow}}}
    \def\rprod{\mathop{\prod{\mkern-28.0mu}{\mathbf\longrightarrow}}}
\begin{document}
\title{ Boxed UC plane partitions and the two-site generalized phase model}
\author{Shengyu Zhang, Jinzhou Liu and  Zhaowen Yan$^*$}

\dedicatory {School of Mathematical Sciences, Inner Mongolia University, \\ Hohhot, Inner Mongolia 010030,  P. R. China }

\thanks{*Corresponding author. Email: yanzw@imu.edu.cn (Z.W. Yan)}

\begin{abstract}
This study investigates the connection between boxed UC plane partitions and the two-site generalized phase model. By introducing two maps, we investigate the representation of  two-side generalized phase algebras and  actions of monodromy matrix operators on  basis vectors. The generating function of boxed UC plane partitions is established by the scalar product of the two-site generalized phase model, which can be expressed as products of Schur functions.
It is shown that the generating function  of boxed UC plane partitions is that of UC plane partitions with the double scaling limit.	
\end{abstract}
\maketitle
\tableofcontents
\section{Introduction}
Symmetric functions\cite{Mac} have wide applications in combinatorics\cite{combinatorics},  representation theory \cite{representation} and  integrable models.
Schur functions and Schur Q-functions are $\tau$ functions of the Kadomtsev-Petviashvili (KP) hierarchy and BKP hierarchy\cite{BKP-1,BKP-2,V.G.,Jimbo1981}, respectively. Hall-Littlewood functions reduce to Schur functions and Schur Q-functions with the parameter $t=0$ and $t=-1$. Jing \cite{Jing-B,Jing-F} investigated the vertex operator expression for the Hall-Littlewood polynomials and  the deformed boson-fermion correspondence. Koike \cite{uc} proposed generalized Schur functions called universal characters (UC) indexed with a pair of partitions.
Tsuda\cite{UC} constructed UC hierarchy whose
$\tau$ functions are UC. The charged fermionic description of  the UC hierarchy are also presented. Meanwhile, Ogawa \cite{BUC} constructed the vertex operator realization of
the generalized Schur Q-function and defined  UC hierarchy of B-type.

Considerable attention has been devoted to studying exactly solvable quantum models by the Quantum Inverse Scattering Method (QISM) \cite{Faddeev,Bogoliubov94}. The phase model is a strongly correlated exactly solvable model characterized by interacting phase operators, which corresponds to the crystal limit of the $q$-boson model \cite{q-boson,phase-Correlators}. By using QISM, Bogoliubov et al. \cite{phase-q} presented the correlation functions of the phase model as determinants and solutions are also provided for two correlation functions of the pahse model \cite{Bogoliubov11}.  Tsilevich \cite{q-boson-hall-little} proved that the wave functions of the phase model and the $q$-boson model can be expressed by the Schur functions and Hall-Littlewood functions, respectively, thereby providing a realization in the algebra of symmetric functions. Furthermore, Wang and Li \cite{two-phase}  proposed the two-site generalized phase model and  the $A$-model topological string partition function. Meanwhile, fermion representations of the two-site generalized phase model have been analyzed \cite{fermion-representation-generalized-phase}.

Plane partitions play a crucial role in statistical physics and quantum field theory, such as in direct percolation \cite{Rajesh98}, topological vertex in local Calabi-Yau geometry \cite{Aganagic}, topological strings on Calabi-Yau threefolds and crystal melting model \cite{Okounkov06}. Much interest has been attributed to generating functions of plane partitions\cite{Andrews,Stanley}. By fermion calculus approach,  Foda et al. \cite{Foda09,Foda07} presented generating functions of KP and BKP plane partitions.  Moreover, the generating function of the boxed plane partition has been derived from the scalar product of the phase model \cite{Boxed-q-boson-model}. Shigechi and Uchiyama \cite{Boxed-skew-phase-model} obtained the generating functions of the boxed skew plane partition by parameterizing variables in the generalized scalar product of the phase model. Recently, by  acting free fermion vertex operators on state vectors, the UC, BUC plane partitions and the corresponding generating functions have been investigated \cite{UC-BUC}.
The purpose of this paper is to discuss the  relations  between boxed UC plane partitions and the two-site generalized phase model.

The paper is organized as follows. In Section $2$, we  review plane partitions, charged fermions and the two-side generalized phase model.
The representation of two-side generalized phase algebras and the action of the monodromy matrix operators on  basis vectors have been studied in Section $3$. In Section $4$, based on the scalar product of the two-site generalized phase model, the generating function of  boxed UC plane partitions can be written as  products  of Schur functions. Section $5$ is devoted to the investigation of the two-side generalized phase model on the infinite lattice limit. Furthermore, the generating function of UC plane partitions has been obtained
with the double scaling limit.

\section{Preliminaries}
 In this section, we mainly retrospect some basic facts about plane partitions, charged fermions and the two-side generalized phase model\cite{UC,two-phase,boshilunwen,UC-BUC}. Meanwhile,  state vectors corresponding to  $2-$partitions in $\bar{\mathcal{F}}_{0,0}$ and $\bar{\mathcal{F}}_{0,0}^*$ are defined.
\subsection{Plane partitions}

The decreasing sequence of non-negative integers $\mu=(\mu_{1},\mu_{2},\ldots)$ is called a partition with the weight $|\mu|=\sum\limits_{i\geq1}\mu_i$.
$m_{i}=m_{i}(\mu)$ denotes the number of times that the element $i$ appears in the partition $\mu$,
and $ l(\mu)$  represents the number of nonzero elements in $\mu$.
 The partition $\mu$ can be rewritten as
\begin{align}
\mu=(1^{m_{1}}2^{m_{2}}...).
\end{align}
For a pair of partitions $\mu$ and $\nu=(\nu_{1},\nu_{2},\ldots)$, the symbol $\mu\succ\nu$ means that $\nu$ interlaces $\mu$,
if partitions $\mu$ and $\nu$ satisfy
\begin{eqnarray}\
\mu_{1}\geq\nu_{1}\geq\mu_{2}\geq\nu_{2}\geq\ldots.
\end{eqnarray}
The $2-$partition $\chi$ is denoted as $(\chi)= (\mu,\nu )$  with the weight $|\chi|=|\mu|+|\nu|$.
Setting $(\bar{\chi})= (\bar{\mu},\bar{\nu} )$,  $(\chi)\succ(\bar{\chi}) $ holds
\begin{align}
	(\chi)\succ(\bar{\chi})  \quad\Longleftrightarrow\quad    \mu\succ\bar{\mu} \quad \mathrm{and} \quad \nu \succ\bar{\nu}.
\end{align}

A plane partition $\pi$ is a set of non-negative integers $\pi_{ij}$ which  satisfies
\begin{align}
	\pi_{ij}\geq\pi_{(i+1)j}, \quad\pi_{ij}\geq\pi_{i(j+1)}, \quad \lim_{i\to\infty}\pi_{ij}=\lim_{j\to\infty}\pi_{ij}=0,\quad \text{for}~i,j\geq1.
\end{align}
Each plane partition can be represented as a composition of specific partitions, denoted as
$\pi=(\ldots,\pi_{-1}, \pi_0,\pi_1,\ldots)$.  $\pi_{i}$  is given by
\begin{eqnarray}
	\pi_i=\begin{cases}(\pi_{1(i+1)},\pi_{2(i+2)},\pi_{3(i+3)},\ldots)&\quad\text{for}~i\geq0,\\(\pi_{(-i+1)1},\pi_{(-i+2)2},\pi_{(-i+3)3},\ldots)&\quad\text{for}~i\leq-1,\end{cases}
\end{eqnarray}
then  the plane partition $\pi$ satisfies
\begin{eqnarray}
	\emptyset=\pi_{-M}\prec\cdots\prec\pi_{-2}\prec\pi_{-1}\prec\pi_{0}\succ\pi_{1}\succ\pi_{2}\succ\cdots\succ\pi_{N}=\emptyset,
\end{eqnarray}
for sufficiently large $M,N\in\mathbb{N}^{+}$ and the weight  $|\pi|=\sum\limits_{i=-M}^{N}|\pi_i|$.
Meanwhile, the UC plane partition $\Pi=(\ldots,\chi_{-1},\chi_0,\chi_1,\ldots)$ holds
\begin{eqnarray}\label{UC-define}
	\emptyset=\chi_{-M}\prec\cdots\prec\chi_{-2}\prec\chi_{-1}\prec\chi_{0}\succ\chi_{1}\succ\chi_{2}\succ\cdots\succ\chi_{N}=\emptyset.
\end{eqnarray}

Introduce  the generating function of UC plane partitions
\begin{align}\label{UC-generating-function}
\sum_{\begin{smallmatrix}\pi^{1} \text{ and }\pi^{2}\text{ are}\\\text{plane partitions}\end{smallmatrix}}
p^{|\pi^{1} |}q^{|\pi^{2}|}
=\prod_{n=1}^\infty\left(\frac1{1-p^n}\right)^n\prod_{m=1}^\infty\left(\frac1{1-q^m}\right)^m.
\end{align}

Let $B(N,L,M)$ denote  a three-dimensional integer lattice of size $N\times L\times M$, given by
\begin{align}
B(N,L,M) = \{(i, j, k) \mid 0 \leq i \leq N, \, 0 \leq j \leq L, \, 0 \leq k \leq M, \text{~and~} i, j, k \in \mathbb{N}\}.
\end{align}
The plane partition limited to an $N\times L\times M$ box  is called a boxed plane partition. The generating function can be expressed as
\begin{align}\label{boxed-generating-function}
	\sum_{\pi\in B(N,L,M)}q^{|\pi|}
	=\prod_{i=1}^{N} \prod_{j=1}^{L} \prod_{k=1}^{M} \frac{1 - q^{i+j+k-1}}{1 - q^{i+j+k-2}}
	=
	\prod_{i=1}^N\prod_{j=1}^L\frac{1-q^{i+j+M-1}}{1-q^{i+j-1}}.
\end{align}

\subsection{Charged fermions and state vectors}
The algebra $\mathcal{A}$ over $\mathbb{C}$ is generated by charged fermions $\psi_{m},\psi_{m}^{*}$, $\phi_{m}$ and $\phi_{m}^{*}$ $ (m,n\in\mathbb{Z}+\frac12)$  satisfying
\begin{align}\label{commutation relations 1.1}
	&[\psi_m,\psi_n]_{+} = [\psi_m^*,\psi_n^*]_{+} = 0,[\psi_m,\psi_n^*]_{+} = \delta_{m+n,0},\\
	&[\phi_m,\phi_n]_{+} = [\phi_m^*,\phi_n^*]_{+} = 0,[\phi_m,\phi_n^*]_{+} = \delta_{m+n,0},\\
	&[\psi_{m},\phi_{n}]=[\psi_{m},\phi_{n}^{*}]=[\psi_{m}^{*},\phi_{n}]=[\psi_{m}^{*},\phi_{n}^{*}]=0,
\end{align}
and $\psi_m^2=\psi_m^{*2}=\phi_m^2=\phi_m^{*2}=0$,  where
\begin{align}
[a, b] = ab - ba, \quad [a, b]_+ = ab + ba, \quad 	\delta_{m+n,0}=\begin{cases}1,&m+n=0,\\0,&\mathrm{otherwise}.\end{cases}
\end{align}
The element $a \in \mathcal{A}$ can be written as
\begin{align}	
	a=\psi_{m_1}\cdots\psi_{m_r}\psi_{n_1}^*\cdots\psi_{n_s}^*\phi_{\tilde{m}^{\prime}_1}\cdots\phi_{\tilde{m}^{\prime}_{i}}\phi_{\tilde{n}^{\prime}_1}^*\cdots\phi_{\tilde{n}^{\prime}_{j}}^*,
\end{align}
and the charge of the fermions is as follows
\begin{center}
	\renewcommand{\dashlinestretch}{30}
	\hspace{0cm}
	\begin{tabular}{cccccc}
		\toprule
		\text{Fermion} & $\psi_n$ & $\psi_n^*$ & $\phi_n$ & $\phi_n^*$ \\
		\midrule
		\text{charge}& $(1,0)$ & $(-1,0)$ & $(0,1)$ & $(0,-1)$\\
		\bottomrule
	\end{tabular}
\end{center}

Introduce the charged fermionic Fock space
\begin{align}
	\bar{\mathcal{F}}&\stackrel{\mathrm{def}}{=}\mathcal{A}\cdot| \mathrm{vac}\rangle=\left\{a| \mathrm{vac}\rangle\mid a\in\mathcal{A}\right\},
\end{align}
where the  vacuum state satisfies
\begin{align}\label{UC-vacuum-state}
	\psi_n|\mathrm{vac}\rangle&=\psi_n^*|\mathrm{vac}\rangle=\phi_n|\mathrm{vac}\rangle=\phi_n^*|\mathrm{vac}\rangle=0\quad\text{for }n>0.
\end{align}
$\bar{\mathcal{F}}_{a_1,a_2}$ refers the Fock subspace with charge $(a_1,a_2)$.
Similarly,  the  dual vacuum state is given by
\begin{align}\label{UC-dual-vacuum-state}
	\langle\mathrm{vac}|\psi_n&=\langle\mathrm{vac}|\psi_n^*=\langle\mathrm{vac}|\phi_n=\langle\mathrm{vac}|\phi_n^*=0\quad\text{for }n<0,
\end{align}
the dual Fock space $\bar{\mathcal{F}}^{*}$ is generated by
\begin{align}
	\bar{\mathcal{F}}^{*}&\stackrel{\mathrm{def}}{=}\langle\mathrm{vac}|\cdot\mathcal{A}=\{\langle\mathrm{vac}|a\mid a\in\mathcal{A}\},
\end{align}
and $\bar{\mathcal{F}}_{a_1,a_2}^{*}$ is the dual Fock subspace with charge $(a_1,a_2)$.

Note
\begin{align}
	\langle l_1, l_2 | = \langle \text{vac} | \Psi_{l_1}^* \Phi_{l_2}^*,\quad
	| l_1, l_2 \rangle = \Psi_{l_1} \Phi_{l_2} | \text{vac} \rangle,
\end{align}
where
\begin{align}
\Psi_l^* &= \begin{cases}
	\psi_{1/2} \cdots \psi_{-l-1/2} & \text{for } l < 0, \\
	1 & \text{for } l = 0, \\
	\psi_{1/2}^* \cdots \psi_{l-1/2}^* & \text{for } l > 0,
\end{cases}\\
\Psi_l &= \begin{cases}
	\psi_{l+1/2}^* \cdots \psi_{-1/2}^* & \text{for } l < 0, \\
	1 & \text{for } l = 0, \\
	\psi_{-l+1/2} \cdots \psi_{-1/2} & \text{for } l > 0,
\end{cases}
\end{align}
and $\Phi_l$ and $\Phi_l^*$ can be obtained by replacing $\psi$ and $\psi^*$ with $\phi$ and $\phi^*$.
For $i=1,2$, the following equations hold
\begin{align}
	\langle l_1, l_2 | \psi_{n_1}&=\langle l_1, l_2 | \phi_{n_2}=0 \quad\text{for } n_{i}<-l_{i},\\
	\psi_{n_1}| l_1, l_2 \rangle&=\phi_{n_2}| l_1, l_2 \rangle=0 \quad\text{for } n_{i}>-l_{i},\\
	\langle l_1, l_2 | \psi_{n_1}^* &=\langle l_1, l_2 | \phi_{n_2}^*=0 \quad\text{for } n_{i}<l_{i},\\
	\psi_{n_1}^*| l_1, l_2 \rangle&=\phi_{n_2}^*| l_1, l_2 \rangle=0 \quad\text{for } n_{i}>l_{i}.
\end{align}

Consider a pairing $\langle\quad\rangle:\bar{\mathcal{F}}^{*}\times\bar{\mathcal{F}}\rightarrow\mathbb{C}$ denoted by
\begin{eqnarray}
	\left(\langle\mathrm{vac}|a,b|\mathrm{vac}\rangle \right) &\longmapsto & \langle\mathrm{vac}|a\cdot b|\mathrm{vac}\rangle=\langle ab\rangle,
\end{eqnarray}
where  $\langle\quad\rangle$ is called the vacuum expectation value.

The following properties hold
\begin{eqnarray}
	\langle\mathrm{vac}|\mathrm{vac} \rangle=1,\quad\langle\psi_m\psi_n^*\rangle=\langle\phi_m\phi_n^*\rangle
	=\begin{cases}\delta_{m+n,0}&(m>0),\\0&(\mathrm{otherwise}).\end{cases}\quad
\end{eqnarray}
Define the colon operator $:~:$ as
\begin{align}\label{colon 1}
	:\psi_m\psi_n^*:=\psi_m\psi_n^*-\langle\psi_m\psi_n^*\rangle,\quad:\phi_m\phi_n^*:=\phi_m\phi_n^*-\langle\phi_m\phi_n^*\rangle.
\end{align}
Consider the operators $H_{n}$ and $\tilde{H}_{n}$ $(n\in\mathbb{Z})$,
\begin{eqnarray}\
	H_{n}=\sum_{j\in\mathbb{Z}+1/2}:\psi_{-j}\psi_{j+n}^{*}:,\quad \tilde{H}_{n}=\sum_{j\in\mathbb{Z}+1/2}:\phi_{-j}\phi_{j+n}^{*}:,
\end{eqnarray}
which hold
\begin{align}\label{commutation relations 2.1}
	[H_{n},\psi_{m}]&=\psi_{m+n},\quad[H_{n},\psi_{m}^{*}]=-\psi_{m+n}^{*},\quad [H_{m},H_{n}]=m\delta_{m+n,0},\\
	[\tilde{H}_{n},\phi_{m}]&=\phi_{m+n},\quad[\tilde{H}_{n},\phi_{m}^{*}]=-\phi_{m+n}^{*},\quad
	[\tilde{H}_{m},\tilde{H}_{n}]=m\delta_{m+n,0},
\end{align}
and
\begin{align}
	H_n|\mathrm{vac} \rangle=\tilde{H}_{n}|\mathrm{vac}\rangle=0\quad\text{if } n>0.
\end{align}

Set $\mu^{j}=(\mu^{j}_{1},\mu^{j}_{2},\ldots)$ and $\nu^{j}=(\nu^{j}_{1},\nu^{j}_{2},\ldots)$. For $l_{1},l_{2}<0$, the right state vector corresponding to the $2-$partition $(\chi)=(\mu^1,\mu^2)$ in the fermionic Fock space $\bar{\mathcal{F}}_{0,0}$ is defined as
\begin{align}\label{right-state}
	| \chi \rangle =| \mu^{1},\mu^{2} \rangle =\psi_{m_1^1} \ldots \psi_{m_{-l_1}^1}\phi_{m_1^2} \ldots \phi_{m_{-l_2}^2}| l_{1},l_{2} \rangle,
\end{align}
where 	$m_{i}^j=-(\mu^{j}_{i} - i)- \frac{1}{2}$ and $m_{1}^j< \cdots < m_{l_j}^j\leq -l_j$ for all $1 \leq i \leq -l_j$, $j=1,2$.
Similarly, setting 	$m_{i}^j=(\mu^{j}_{i} - i)+\frac{1}{2}$,
the identified left state vector in $\bar{\mathcal{F}}_{0,0}^*$ is represented as
\begin{align}\label{left-state}
	\langle	 \chi | =\langle \mu^2,\mu^1|=\langle l_2,l_1|  \phi_{m_{-l_2}^2}^* \ldots \phi_{m_1^2}^*  \psi_{m_{-l_1}^1}^* \ldots \psi_{m_1^1}^*,
\end{align}
where $m_{1}^j>\cdots> m_{l_j}^j\geq l_j$  for all  $1 \leq i \leq -l_j$, $j=1,2$.
Note that $n_{i}^j=-(\nu^{j}_{i} - i)- \frac{1}{2}$, $m_{i}^j=(\mu^{j}_{i} - i)+\frac{1}{2}$ and $l_{j}, k_{j}<0$, then
\begin{align}\label{fock-scalar-product}
\big( \langle \mu^2,\mu^1|,| \nu^{1},\nu^{2} \rangle  \big)
=&\langle l_2,l_1|  \phi_{m_{-l_2}^2}^* \ldots \phi_{m_1^2}^*  \psi_{m_{-l_1}^1}^* \ldots \psi_{m_1^1}^*
\psi_{n_1^1} \ldots \psi_{n_{-k_1}^1}\phi_{n_1^2} \ldots \phi_{n_{-k_2}^2}| k_{1},k_{2} \rangle \notag\\
=&\delta_{k_{1},l_{1}} \delta_{k_{2},l_{2}}
\prod_{i=1}^{-l_{1}} \delta_{m_{i}^1,n_{i}^1}
\prod_{j=1}^{-l_{2}} \delta_{m_{j}^2,n_{j}^2}
= \delta_{\mu^1,\nu^1}\delta_{\mu^2,\nu^2},
\end{align}
where $\delta_{\mu^j,\nu^j}$  denotes that $\mu^j$ equals $\nu^j$.



\subsection{The two-site generalized phase model}

The two-site generalized phase algebra  is generated by  $\{ \phi^{(j)}, \phi^{(j)\dagger}, \mathcal{N}^{(j)}, \pi^{(j)}, j=1,2 \}$,  which satisfies
\begin{align}
&[\phi^{(i)},\phi^{(j)\dagger}]=\delta_{i,j}\pi^{(i)},
\quad[\mathcal{N}^{(i)},\phi^{(j)}]=-\delta_{i,j}\phi^{(i)},
\quad[\mathcal{N}^{(i)},\phi^{(j)\dagger}]=\delta_{i,j}\phi^{(i)\dagger},\\
&\phi^{(i)}\pi^{(i)}=\pi^{(i)}\phi^{(i)\dagger}=0,\quad \mathrm{for~}i,j=1,2.
\end{align}
Consider $M_{1}+M_{2}+2$ copies of the two-site generalized phase algebra
represented by $\{ \phi_k^{(j)}, \phi_k^{(j)\dagger}, \mathcal{N}_k^{(j)}, \\ \pi_k^{(j)} \},$ $0\leq k\leq M_j$ .
The algebra between different types and positions is commutative, leading to
\begin{align}\label{generalized-calculate}
	[\phi_{l_{i}}^{(i)},\phi_{k_j}^{(j)\dagger}]&=
\left\{\begin{array}{ll}
	\delta_{l_{i},k_j}\pi_{l_{i}}^{(j)}, &  i=j,
	\\  0,&  i \neq j,
\end{array}
\right.	
	\quad	
	[\mathcal{N}_{l_{i}}^{(i)},\phi_{k_j}^{(j)}]=
	\left\{\begin{array}{ll}
		-\delta_{l_{i},k_{j}}\phi_{l_{i}}^{(j)}, &  i=j,
		\\  0,&  i \neq j,
	\end{array}
	\right.
\\
	[\mathcal{N}_{l_{i}}^{(i)},\phi_{k_j}^{(j)\dagger}]&=
	\left\{\begin{array}{ll}
		\delta_{l_{i},k_{j}}\phi_{l_{i}}^{(j)\dagger}, &  i=j,
		\\  0,&  i \neq j,
	\end{array}
	\right.
	\ \  \phi_{l_{i}}^{(i)}\pi_{l_{i}}^{(i)}=\pi_{l_{i}}^{(i)}\phi_{l_{i}}^{(i)\dagger}=0,
\end{align}
where $0\leq l_{i}, k_{j} \leq M_j$, and $i,j=1,2$.
The $M_1 +1$ and $M_2 +1$ copies of the vector space $\mathcal{V}$ are written as
\begin{align}
\mathcal{V}^{(1)}=\bigotimes_{i=0}^{M_{1}}\mathcal{V}_{i}^{(1)},\quad
\mathcal{V}^{(2)}=\bigotimes_{j=0}^{M_{2}}\mathcal{V}_{j}^{(2)},
\end{align}
where
\begin{align}
	\text{Basis}(\mathcal{V}^{(1)})&=  \left\{|n_{0}^{1}\rangle_0^{(1)}\otimes\cdots\otimes|n_{M_1}^{1}\rangle_{M_1}^{(1)}                  =	\bigotimes_{i=0}^{M_{1}}|n_{i}^{1} \rangle_{i}^{(1)} \right\},\\
\text{Basis}(\mathcal{V}^{(2)})&= \left\{|n_{0}^{2}\rangle_0^{(2)}\otimes\cdots\otimes|n_{M_2}^{2}\rangle_{M_2}^{(2)}      =	\bigotimes_{j=0}^{M_{2}}|n_{j}^{2} \rangle_{j}^{(2)} \right\},
\end{align}
and  $n_i^{j}$ $(j=1,2$ and $0\leq i\leq M_j)$ are called the occupation numbers of $i^{th}$ position of the basis vector and range over all non-negative integers.
We denote  the number of particles as $N_j=\sum\limits_{i=0}^{M_j} n_i^j$.
Then basis vectors of the vector space $\mathcal{V}=\mathcal{V}^{(1)}\bigotimes\mathcal{V}^{(2)}$ is of the form
\begin{align}
\text{Basis}(\mathcal{V})&=
	\left\{|n^{1}\rangle^{(1)}\bigotimes|n^{2}\rangle^{(2)}=	\bigotimes_{i=0}^{M_{1}}|n^{1}_{i}\rangle_{i}^{(1)}\bigotimes\bigotimes_{j=0}^{M_{2}}|n^{2}_{j}\rangle_{j}^{(2)}  \right\}.
\end{align}
The inner product between two basis vectors $|m^{1}\rangle^{(1)}\bigotimes|m^{2}\rangle^{(2)}$  and
$|n^{1}\rangle^{(1)}\bigotimes|n^{2}\rangle^{(2)}$ is defined as
\begin{align}
	\mathcal{I}\left(|m^{1}\rangle^{(1)}\bigotimes|m^{2}\rangle^{(2)}, |n^{1}\rangle^{(1)}\bigotimes|n^{2}\rangle^{(2)}\right) = \prod_{i=1}^{M_{1}} \delta_{m_{i}^1,n_{i}^1}
	\prod_{j=1}^{M_{2}} \delta_{m_{j}^2,n_{j}^2}.
\end{align}
Similarly, the corresponding dual vector space is denoted as $\mathcal{V}^{\ast}=\mathcal{V}^{(1)\ast}\bigotimes\mathcal{V}^{(2)\ast}$ and
\begin{align}
\text{Basis}(\mathcal{V}^{\ast})
	=\left\{
	\prescript{(2)}{}{\langle m^{2}|}\bigotimes \prescript{(1)}{}{\langle m^{1}|}=
	\bigotimes_{j=0}^{M_{2}}\prescript{(2)}{j}{\langle m^{2}_j|}\bigotimes
	\bigotimes_{i=0}^{M_{1}}\prescript{(1)}{i}{\langle m^{1}_i|}
	\right\}.
\end{align}

Introduce the $R$-matrix for the two-side generalized phase model
\begin{align}
R(x,y) = \begin{pmatrix}
	x & 0 & 0 & 0 \\
	0 & 0 & x^{\frac{1}{2}} y^{\frac{1}{2}} & 0 \\
	0 & x^{\frac{1}{2}} y^{\frac{1}{2}} & x - y & 0 \\
	0 & 0 & 0 & x
\end{pmatrix},
\end{align}
the $L-$matrices
\begin{align}\label{L-j}
	L_{i}^{(j)}(x)=
	\begin{pmatrix}
		x^{-\frac{1}{2}} &  \phi_i^{(j)\dagger} \\
		\phi_i^{(j)} & x^{\frac{1}{2}}
	\end{pmatrix}, \quad j=1,2\mathrm{~and~}i=0,\cdots,M_j,
\end{align}
and the monodromy matrix		
\begin{align}
	T(x)=T_{2}(x) \cdot T_{1}(x),
\end{align}
where	
\begin{align}\label{T-j}
	\begin{aligned}
		T_{j}(x) =L_{M_{j}}^{(j)}(x)L_{M_{j}-1}^{(j)}(x)\cdots L_{0}^{(j)}(x)
		=
		\begin{pmatrix}
			A_j(x) & B_j(x) \\
			C_j(x) & D_j(x)
		\end{pmatrix},
	\end{aligned}\quad j=1,2.
\end{align}

Define  operators
\begin{align}
\hat{\mathcal{N}}_1 = \sum_{i=0}^{M_1} \mathcal{N}_i^{(1)}, \quad \hat{\mathcal{N}}_2 = \sum_{i=0}^{M_2} \mathcal{N}_i^{(2)}, \quad \text{and} \quad \hat{\mathcal{N}} = \hat{\mathcal{N}}_1 + \hat{\mathcal{N}}_2.
\end{align}
The operators $A_j(x)$ and $D_j(x)$ do not change the total number of particles. However,  operators $B_j(x)$ and $C_j(x)$ satisfy
\begin{align}
\hat{\mathcal{N}}_j B_j(u) = B_j(u) (\hat{\mathcal{N}}_j + 1), \quad \hat{\mathcal{N}}_j C_j(u) = C_j(u) (\hat{\mathcal{N}}_j - 1),
\end{align}
where we refer to $B_j(x)$ as the creation operators and $C_j(x)$ as the annihilation operators.

The following bilinear equations  hold
\begin{align}
	R(x,y)(T(x)\otimes T(y))&=(T(y)\otimes T(x))R(x,y),\\
	R(x,y)(T_{j}(x)\otimes T_{j}(y))&=(T_{j}(y)\otimes T_{j}(x))R(x,y),\\
	R(x,y)(L_{i}^{(j)}(x)\otimes L_{i}^{(j)}(y))&=(L_{i}^{(j)}(y)\otimes L_{i}^{(j)}(x))R(x,y).\label{bilinear-3}
\end{align}
From the Eq.$(\ref{bilinear-3})$, one can obtain
\begin{align}\label{B,C,commute}
	[B_i(x),B_j(y)]=[C_i(x),C_j(y)]=0, \quad \mathrm{for~}i,j=1,2.
\end{align}		

\section{The action of the monodromy matrix operators on basis vectors}
This section provides the representation of two-side generalized phase algebras and we shall calculate the scalar product of basis vectors.
By using the maps $\mathcal{M}$ and $\mathcal{M}^*$,
we show that the action of monodromy matrix operators on the basis vectors generates interlacing boxed $2-$partitions.
\subsection{The representation of two-side   generalized phase algebras}
Consider the following action of operators $\{ \phi_i^{(j)}, \phi_i^{(j)\dagger}, \mathcal{N}_i^{(j)}, \pi_i^{(j)} \}$ $(j=1,2$ and $0\leq i\leq M_j)$ on the basis vector $\bigotimes\limits_{i=0}^{M_{1}}|n_{i}^{1} \rangle_{i}^{(1)}\bigotimes\bigotimes\limits_{j=0}^{M_{2}}|n_{j}^{2}\rangle_{j}^{(2)} $ in the vector space $\mathcal{V}$.

$\phi_i^{(1)}$ annihilates the basis vector or reduces the $i^{th}$ occupation number of the basis vector by $1$ with the relation
\begin{align}\label{right-ym}
	&\phi_i^{(1)} \bigotimes_{i=0}^{M_{1}}|n_{i}^{1} \rangle_{i}^{(1)}\bigotimes\bigotimes_{j=0}^{M_{2}}|n_{j}^{2}\rangle_{j}^{(2)}   \notag\\
	=&\left\{
	\begin{array}
		{ll}0, & n_{i}^{1}=0, \\
		\\
		|n_{0}^{1} \rangle_0^{(1)}\otimes\cdots\otimes|n_{i}^{1}-1\rangle_i^{(1)}\otimes\cdots\otimes|n_{M_1}^{1}\rangle_{M_1}^{(1)}\bigotimes\bigotimes\limits_{j=0}^{M_{2}}|n_{j}^{2}\rangle_{j}^{(2)}, & n_{i}^{1}\geq1.
	\end{array}\right.
\end{align}	
$\phi_i^{(1)\dagger}$ increases the $i^{th}$ occupation number of the basis vector by $1$ with
\begin{align}\label{right-sc}
	\phi_i^{(1)\dagger}
	\bigotimes_{i=0}^{M_{1}}|n_{i}^{1} \rangle_{i}^{(1)}\bigotimes\bigotimes_{j=0}^{M_{2}}|n_{j}^{2}\rangle_{j}^{(2)} 		
	=|n_{0}^{1}\rangle_0^{(1)}\otimes\cdots\otimes|n_{i}^{1}+1\rangle_i^{(1)}\otimes\cdots\otimes|n_{M_1}^{1}\rangle_{M_1}^{(1)}\bigotimes\bigotimes\limits_{j=0}^{M_{2}}|n_{j}^{2}\rangle_{j}^{(2)}.
\end{align}	
The basis vector is the eigenvector of $\mathcal{N}_i^{(1)}$ with the $i^{th}$ occupation number as its eigenvalue,
\begin{align}\label{right-eigenvalue}
	\mathcal{N}_i^{(1)}
	\bigotimes_{i=0}^{M_{1}}|n_{i}^{1} \rangle_{i}^{(1)}\bigotimes\bigotimes_{j=0}^{M_{2}}|n_{j}^{2}\rangle_{j}^{(2)}  	
	=n_{i}^{1}
\bigotimes_{i=0}^{M_{1}}|n_{i}^{1} \rangle_{i}^{(1)}\bigotimes\bigotimes_{j=0}^{M_{2}}|n_{j}^{2}\rangle_{j}^{(2)} .
\end{align}	
The action of $\pi_i^{(1)} $ on the basis vector is given by
\begin{align}\label{right-pi}
	\pi_i^{(1)}
	\bigotimes_{i=0}^{M_{1}}|n_{i}^{1}\rangle_{i}^{(1)}\bigotimes\bigotimes_{j=0}^{M_{2}}|n_{j}^{2}\rangle_{j}^{(2)}  		
	=\left\{
	\begin{array}
		{ll}|n_{0}^{1} \rangle_0^{(1)}\otimes\cdots\otimes|n_{M_1}^{1}\rangle_{M_1}^{(1)}\bigotimes\bigotimes\limits_{j=0}^{M_{2}}|n_{j}^{2}\rangle_{j}^{(2)}  , & n_{i}^{1}=0, \\
		\\
		0, & n_{i}^{1}\geq1.
	\end{array}\right.
\end{align}	
Meanwhile, the actions of $\phi_i^{(1)}$ and $\phi_i^{(1)\dagger}$ get interchanged in the dual vector space $\mathcal{V}^{\ast}$ with the following relations
\begin{align}	
&\bigotimes_{j=0}^{M_{2}}\prescript{(2)}{j}{\langle m^{2}_j|}\bigotimes
\bigotimes_{i=0}^{M_{1}}\prescript{(1)}{i}{\langle m^{1}_i|}
\phi_i^{(1)}\notag\\
=&	\bigotimes\limits_{j=0}^{M_{2}}\prescript{(2)}{j}{\langle m^{2}_j|}\bigotimes	
\prescript{(1)}{0}{\langle m^{1}_0|}\otimes\cdots\otimes
\prescript{(1)}{i}{\langle m^{1}_{i}+1|}
\otimes\cdots\otimes\prescript{(1)}{M_1}{\langle m^{1}_{M_1}|},\label{left-sc}
\end{align}
and
\begin{align}
	&\bigotimes_{j=0}^{M_{2}}\prescript{(2)}{j}{\langle m^{2}_j|}\bigotimes
	\bigotimes_{i=0}^{M_{1}}\prescript{(1)}{i}{\langle m^{1}_i|}
	\phi_i^{(1)\dagger} \notag\\
	=&\left\{
	\begin{array}
		{ll}0, & m^{1}_i=0, \\
		\\
	\bigotimes\limits_{j=0}^{M_{2}}\prescript{(2)}{j}{\langle m^{2}_j|}\bigotimes	
		\prescript{(1)}{0}{\langle m^{1}_0|}\otimes\cdots\otimes
		\prescript{(1)}{i}{\langle m^{1}_{i}-1|}
		\otimes\cdots\otimes\prescript{(1)}{M_1}{\langle m^{1}_{M_1}|},
		& m^{1}_i\geq1.
	\end{array}\right.\label{left-ym}
\end{align}	
 The action of   $\mathcal{N}_i^{(1)}$ and $\pi_i^{(1)} $  on the dual vector space $\mathcal{V}^{\ast}$ are given by
\begin{align}
	&\bigotimes_{j=0}^{M_{2}}\prescript{(2)}{j}{\langle m^{2}_j|}\bigotimes
	\bigotimes_{i=0}^{M_{1}}\prescript{(1)}{i}{\langle m^{1}_i|}
	\mathcal{N}_i^{(1)}
	=m^{1}_i
	\bigotimes\limits_{j=0}^{M_{2}}\prescript{(2)}{j}{\langle m^{2}_j|}\bigotimes
	\prescript{(1)}{0}{\langle m^{1}_0|}\otimes\cdots\otimes\prescript{(1)}{M_1}{\langle m^{1}_{M_1}|},\label{left-eigenvalue}\\
&\bigotimes_{j=0}^{M_{2}}\prescript{(2)}{j}{\langle m^{2}_j|}\bigotimes
\bigotimes_{i=0}^{M_{1}}\prescript{(1)}{i}{\langle m^{1}_i|}
\pi_i^{(1)}
=\left\{
\begin{array}
	{ll}	\bigotimes\limits_{j=0}^{M_{2}}\prescript{(2)}{j}{\langle m^{2}_j|}\bigotimes
	\prescript{(1)}{0}{\langle m^{1}_0|}\otimes\cdots\otimes\prescript{(1)}{M_1}{\langle m^{1}_{M_1}|}, & m^{1}_i=0, \\
	\\
	0, & m^{1}_i\geq1.
\end{array}\right.	\label{left-pi}	
\end{align}	
It should be pointed out that
$\{ \phi_i^{(2)}, \phi_i^{(2)\dagger}, \mathcal{N}_i^{(2)}, \pi_i^{(2)} \}$
and
$\{ \phi_i^{(1)}, \phi_i^{(1)\dagger}, \mathcal{N}_i^{(1)}, \pi_i^{(1)} \}$
act on the basis vectors with the  similar rules. Hence, we will not  repeat them here.	
\begin{lemma}
The action  of the basis vectors $\langle m^{2}|^{(2)}\bigotimes \langle m^{1}|^{(1)}\in \mathcal{V}^{\ast}$ on $|n^{1}\rangle^{(1)}\bigotimes|n^{2}\rangle^{(2)} \in \mathcal{V} $
satisfies
\begin{align}\label{tgpm-scalar-product}
	\prescript{(2)}{}{\langle m^{2}|}\bigotimes \prescript{(1)}{}{\langle m^{1}|}
	n^{1}\rangle^{(1)}\bigotimes|n^{2}\rangle^{(2)}
	= \prod_{j=1}^{2}\prod_{i=0}^{M_j} \delta_{m_i^j, n_i^j}.
\end{align}
\end{lemma}

\begin{proof}
By means of the Eq.$(\ref{right-sc})$, one can obtain
\begin{align}
	&|n^{1}\rangle^{(1)}\bigotimes|n^{2}\rangle^{(2)}\notag\\
	=&
	\bigotimes_{i=0}^{M_{1}}|n^{1}_{i}\rangle_{i}^{(1)}\bigotimes\bigotimes_{j=0}^{M_{2}}|n^{2}_{j}\rangle_{j}^{(2)} \nonumber \\
	=&(\phi_{0}^{(1)\dagger})^{n_{0}^{1}}\cdots(\phi_{M_1}^{(1)\dagger})^{n_{M_1}^{1}}
	(\phi_{0}^{(2)\dagger})^{n_{0}^{2}}\cdots(\phi_{M_2}^{(2)\dagger})^{n_{M_2}^{2}}
	|0\rangle_0^{(1)}\otimes\cdots\otimes|0\rangle_{M_1}^{(1)}\otimes
	|0\rangle_0^{(2)}\otimes\cdots\otimes|0\rangle_{M_2}^{(2)}\notag\\
	=&\lprod_{l=0}^{M_1}\phi_{l}^{(1)\dagger} \lprod_{k=0}^{M_2}\phi_{k}^{(2)\dagger}
	|0\rangle^{(1)}\bigotimes|0\rangle^{(2)},
\end{align}
where the indices decrease in the direction of the arrow in the product symbol.
Similarly, from the Eq.$(\ref{left-sc})$, we have
\begin{align}
	\prescript{(2)}{}{\langle m^{2}|}\bigotimes \prescript{(1)}{}{\langle m^{1}|}=\prescript{(2)}{}{\langle 0|}\bigotimes \prescript{(1)}{}{\langle0|}
	\lprod_{s=0}^{M_1}\phi_{s}^{(1)} \lprod_{t=0}^{M_2}\phi_{t}^{(2)}.
\end{align}
Then
\begin{align}
	&\prescript{(2)}{}{\langle m^{2}|}\bigotimes \prescript{(1)}{}{\langle m^{1}|}
	n^{1}\rangle^{(1)}\bigotimes|n^{2}\rangle^{(2)}\notag\\
	= &\prescript{(2)}{}{\langle 0|}\bigotimes \prescript{(1)}{}{\langle0|}
	\lprod_{s=0}^{M_1}\phi_{s}^{(1)} \lprod_{t=0}^{M_2}\phi_{t}^{(2)}
	\lprod_{l=0}^{M_1}\phi_{l}^{(1)\dagger} \lprod_{k=0}^{M_2}\phi_{k}^{(2)\dagger}
	|0\rangle^{(1)}\bigotimes|0\rangle^{(2)}.
\end{align}
Combining the Eq.(\ref{generalized-calculate}) and the fact $\phi_i^{(j)} |0\rangle = \langle 0| \phi_i^{({j})\dagger} = 0$ yields
\begin{align}
	\prescript{(2)}{}{\langle m^{2}|}\bigotimes \prescript{(1)}{}{\langle m^{1}|}
	n^{1}\rangle^{(1)}\bigotimes|n^{2}\rangle^{(2)}
	= \prod_{j=1}^{2}\prod_{i=0}^{M_j} \delta_{m_i^j, n_i^j}.
\end{align}
\end{proof}


\subsection{Generating interlacing boxed $2-$partitions}
Set $| n^{1}\rangle^{(1)}\bigotimes|n^{2}\rangle^{(2)}$ and $\prescript{(2)}{}{\langle n^{2}|}\bigotimes \prescript{(1)}{}{\langle n^{1}|}$ are  basis vectors in $\mathcal{V}$ and $\mathcal{V}^*$, respectively.
\begin{definition}\label{M}
Define the linear maps
\begin{align}
	\mathcal{M} : \mathcal{V} \to \bar{\mathcal{F}}_{0,0}, \quad
	\mathcal{M}^* : \mathcal{V}^* \to \bar{\mathcal{F}}_{0,0}^{*},
\end{align}
which follow
\begin{align}\label{M-map}
	\mathcal{M}  | n^{1}\rangle^{(1)}\bigotimes|n^{2}\rangle^{(2)}&=| \nu^{1},\nu^{2} \rangle =| \chi \rangle,\\
	\prescript{(2)}{}{\langle n^{2}|}\bigotimes \prescript{(1)}{}{\langle n^{1}|}  \mathcal{M}^* &= \langle \nu^{2},\nu^{1}| =\langle \chi|,
\end{align}
where
\begin{align}
	\nu^{j}=(0^{n^{j}_0}  1^{n^{j}_1}  \cdots  M_{j}^{n^{j}_{M_{j}}})=(1^{n^{j}_1}  \cdots  M_{j}^{n^{j}_{M_{j}}}),\quad j=1,2.
\end{align}	
\end{definition}

For all $\prescript{(2)}{}{\langle m^{2}|}\bigotimes \prescript{(1)}{}{\langle m^{1}|} \in \mathcal{V}^*$ and $|n^{1}\rangle^{(1)}\bigotimes|n^{2}\rangle^{(2)} \in \mathcal{V}$ satisfying $m^{j}_0=n^{j}_0$,
as can be seen from  Eqs.$(\ref{fock-scalar-product})$ and $(\ref{tgpm-scalar-product})$ that
\begin{align}\label{scalar-product-equal}
\prescript{(2)}{}{\langle m^{2}|}\bigotimes \prescript{(1)}{}{\langle m^{1}|}
n^{1}\rangle^{(1)}\bigotimes|n^{2}\rangle^{(2)}
=\big(\prescript{(2)}{}{\langle m^{2}|}\bigotimes \prescript{(1)}{}{\langle m^{1}|}\mathcal{M}^*,\mathcal{M}| n^{1}\rangle^{(1)}\bigotimes|n^{2}\rangle^{(2)} \big).
\end{align}
For basis vectors $| m^{1}\rangle^{(1)}\bigotimes|m^{2}\rangle^{(2)}$ and $| n^{1}\rangle^{(1)}\bigotimes|n^{2}\rangle^{(2)}$ of the space $\mathcal{V}$,
let us denote
\begin{align}
	\Sigma_i^{m^j} = \sum_{k=i}^{M_j} m_k^j, \quad \Sigma_i^{n^j} = \sum_{k=i}^{M_j} n_k^j,\quad j=1,2  \text{~and~} 0\leq i \leq M_{j}.
\end{align}
The symbol $| m^{j}\rangle^{(j)} \triangleright   |   n^{j}\rangle^{(j)}$ means that $| m^{j}\rangle^{(j)}$ is admissible to   $ | n^{j}\rangle^{(j)}$, if and only if
\begin{align}\label{admissible-condition}
	(\Sigma_0^{m^j} - \Sigma_0^{n^j}) = 1\quad \text{and} \quad 0 \leq (\Sigma_i^{m^j} - \Sigma_i^{n^j}) \leq 1, \quad j=1,2  \text{~and~} 0\leq i \leq M_{j}.
\end{align}
Similarly,
we say that  $\prescript{(j)}{}{\langle m^{j}|}$  is admissible to $\prescript{(j)}{}{\langle n^{j}|}$ and write $ \prescript{(j)}{}{\langle n^{j}|}  \triangleleft  \prescript{(j)}{}{\langle m^{j}|}$, if basis vectors $\prescript{(2)}{}{\langle m^{2}|}\bigotimes \prescript{(1)}{}{\langle m^{1}|}$ and $\prescript{(2)}{}{\langle n^{2}|}\bigotimes \prescript{(1)}{}{\langle n^{1}|} $ of the space $\mathcal{V}^*$ satisfy the Eq.(\ref{admissible-condition}).
It is noted that
\begin{align}\label{admiss}
&| m^{1}\rangle^{(1)}\bigotimes|m^{2}\rangle^{(2)} \triangleright  | n^{1}\rangle^{(1)}\bigotimes|n^{2}\rangle^{(2)}
\Longleftrightarrow
| m^{1}\rangle^{(1)} \triangleright   | n^{1}\rangle^{(1)} \text{ and }
| m^{2}\rangle^{(2)} \triangleright   |   n^{2}\rangle^{(2)},\\
&\prescript{(2)}{}{\langle n^{2}|}\bigotimes \prescript{(1)}{}{\langle n^{1}|}
\triangleleft
\prescript{(2)}{}{\langle m^{2}|}\bigotimes \prescript{(1)}{}{\langle m^{1}|}
\Longleftrightarrow
 \prescript{(1)}{}{\langle n^{1}|}  \triangleleft  \prescript{(1)}{}{\langle m^{1}|}
\text{ and }
\prescript{(2)}{}{\langle n^{2}|}  \triangleleft  \prescript{(2)}{}{\langle m^{2}|}.
\end{align}

\begin{lemma}\label{xr-jc-1}
	If the maps satisfy the relations
	\begin{align}
&|\mu^{1}, \mu^{2}\rangle = \mathcal{M} |m^{1}\rangle^{(1)}\bigotimes|m^{2}\rangle^{(2)},
\quad | \nu^{1},\nu^{2} \rangle = \mathcal{M} | n^{1}\rangle^{(1)}\bigotimes|n^{2}\rangle^{(2)},\\
&\prescript{(2)}{}{\langle m^{2}|}\bigotimes \prescript{(1)}{}{\langle m^{1}|}\mathcal{M}^* = \langle\mu^{2},\mu^{1}| ,
 \quad  \prescript{(2)}{}{\langle n^{2}|}\bigotimes \prescript{(1)}{}{\langle n^{1}|} \mathcal{M}^* = \langle \nu^{2},\nu^{1} |,
	\end{align}
then we have
\begin{align}
	&| m^{1}\rangle^{(1)}\bigotimes|m^{2}\rangle^{(2)} \triangleright  | n^{1}\rangle^{(1)}\bigotimes|n^{2}\rangle^{(2)}
	\implies
	\mu^{1}  \succ  \nu^{1} ,\mu^{2}  \succ  \nu^{2} , \label{r-admiss-interlace}\\
	&\prescript{(2)}{}{\langle n^{2}|}\bigotimes \prescript{(1)}{}{\langle n^{1}|}
	 \triangleleft
	 \prescript{(2)}{}{\langle m^{2}|}\bigotimes \prescript{(1)}{}{\langle m^{1}|}
	\implies  \nu^{1}  \prec  \mu^{1} , \nu^{2} \prec  \mu^{2}.	\label{l-admiss-interlace}
\end{align}	
\end{lemma}
\begin{proof}
From the Lemma $1$ in the Section $3.2.4$ of the reference \cite{boshilunwen}, we have
\begin{align}
| m^{1}\rangle^{(1)} \triangleright	|n^{1}\rangle^{(1)}
\implies \mu^{1}  \succ  \nu^{1},\\
|m^{2}\rangle^{(2)}  \triangleright |n^{2}\rangle^{(2)}
\implies \mu^{2}  \succ  \nu^{2}.
\end{align}
The Eq.$(\ref{admiss})$ implies that the Eq.$(\ref{r-admiss-interlace})$ holds.
Using the similar procedure, the Eq.$(\ref{l-admiss-interlace})$ can also be proved.

\end{proof}

\begin{lemma}\label{B-n,C-n,2}
For $j=1,2$, setting $\bar{\mathbb{B}}_{j}(x) = x^{\frac{M_{j}}{2}} B_{j}(x)$ and $\bar{\mathbb{C}}_{j}(x) = x_{j}^{\frac{M_{j}}{2}} C(\frac{1}{x_{j}})$, the actions of  $\bar{\mathbb{B}}_{j}(x)$ and $\bar{\mathbb{C}}_{j}(x)$ on
$| n^{1}\rangle^{(1)}\bigotimes|n^{2}\rangle^{(2)}$  and $ \prescript{(2)}{}{\langle n^{2}|}\bigotimes \prescript{(1)}{}{\langle n^{1}|}$  satisfy
	\begin{align}
\bar{\mathbb{B}}_{1}(x)\bar{\mathbb{B}}_{2}(y)|n^{1}\rangle^{(1)}\bigotimes|n^{2}\rangle^{(2)}
&= \sum_{{ |m^{1} \rangle^{(1)} \triangleright | n^{1} \rangle^{(1)} }\atop{|m^{2} \rangle^{(2)} \triangleright | n^{2} \rangle^{(2)} }}
	 \prod_{i=1}^{M_1}   x^{i(m^{1}_i - n^{1}_i)} \prod_{j=1}^{M_2}	 y^{i(m^{2}_j - n^{2}_j)}
		| m^{1}\rangle^{(1)}\bigotimes|m^{2}\rangle^{(2)}, \label{1-B-n,C-n,2}  \\
\prescript{(2)}{}{\langle n^{2}|}\bigotimes \prescript{(1)}{}{\langle n^{1}|}
\bar{\mathbb{C}}_{2}(y)\bar{\mathbb{C}}_{1}(x)
 &= \sum_{{ \prescript{(1)}{}{\langle n^{1}|} \triangleleft \prescript{(1)}{}{\langle m^{1}|}  }\atop{\prescript{(2)}{}{\langle n^{2}|}  \triangleleft \prescript{(2)}{}{\langle m^{2}|} }}
\prod_{i=1}^{M_1}   x^{i(m^{1}_i - n^{1}_i)} \prod_{j=1}^{M_2}	 y^{j(m^{2}_j - n^{2}_j)}
 \prescript{(2)}{}{\langle m^{2}|}\bigotimes \prescript{(1)}{}{\langle m^{1}|} .\label{2-B-n,C-n,2}
	\end{align}
\end{lemma}
\begin{proof}
We only prove the Eq.$(\ref{1-B-n,C-n,2})$. The proof for the Eq.$(\ref{2-B-n,C-n,2})$ is similar, so it will not be repeated.
Set $\begin{pmatrix}
	1 \\
	0
\end{pmatrix}= \uparrow$,
$\begin{pmatrix}
	0 \\
	1
\end{pmatrix}= \downarrow$ and  $( 1\ 0 )=\uparrow^*$.
From the Eq.$(\ref{T-j})$, one gets
\begin{align}\label{B1-B2}
B_1 (x) B_2 (y)=&( 1\ 0 )
\begin{pmatrix}
	A_1(x) & B_1(x) \\
	C_1(x) & D_1(x)
\end{pmatrix}
\begin{pmatrix}
	0 \\
	1
\end{pmatrix}
( 1\ 0 )
\begin{pmatrix}
	A_2(y) & B_2(y) \\
	C_2(y) & D_2(y)
\end{pmatrix}
\begin{pmatrix}
	0 \\
	1
\end{pmatrix}\notag\\
=& \uparrow^*  L_{M_1}^{(1)}(x)\cdots L_0^{(1)}(x)\downarrow \uparrow^* L_{M_2}^{(2)}(y)\cdots L_0^{(2)}(y) \downarrow.
\end{align}	
Thus
\begin{align}
	&\prescript{(2)}{}{\langle m^{2}|}\bigotimes \prescript{(1)}{}{\langle m^{1}|}
B_1 (x) B_2 (y)
|n^{1}\rangle^{(1)}\bigotimes|n^{2}\rangle^{(2)}\notag\\
=&\uparrow^*
 \bar{L}_{M_1}^{(1)}(x)\cdots \bar{L}_0^{(1)}(x)\downarrow \uparrow^* \bar{L}_{M_2}^{(2)}(y)\cdots \bar{L}_0^{(2)}(y)
\downarrow,
\end{align}
where
\begin{align}
\bar{L}_{i}^{(j)}(x) =
\begin{pmatrix}
\prescript{(j)}{i}{\langle m^{j}_i|}  x^{- \frac{1}{2}} | n_i^{j} \rangle_i^{(j)} & \prescript{(j)}{i}{\langle m^{j}_i|} \phi_i^{(j)\dagger }| n_i^{j} \rangle_i^{(j)} \\
\prescript{(j)}{i}{\langle m^{j}_i|} \phi_i^{(j)} | n_i^{j} \rangle_i^{(j)} &\prescript{(j)}{i}{\langle m^{j}_i|} x^{\frac{1}{2}}
	| n_i^{j} \rangle_i^{(j)}
\end{pmatrix} \quad j=1,2  \text{~and~} 0\leq i \leq M_{j}.
\end{align}
Combining Eqs.$(\ref{right-ym})$, $(\ref{right-sc})$ and $(\ref{tgpm-scalar-product})$ yields
\begin{align}
\bar{L}_{i}^{(j)}(x) =
\begin{cases}
	\begin{pmatrix}
		x^{-\frac{1}{2}} & 0 \\
		0 & x^{\frac{1}{2}}
	\end{pmatrix}, & m_i^{j} = n_i^{j}, \\
	A_{12}=\begin{pmatrix}
		0 & 1 \\
		0 & 0
	\end{pmatrix}, & m^{j}_i = n_i^{j} + 1, \\
	A_{21}=\begin{pmatrix}
		0 & 0 \\
		1 & 0
	\end{pmatrix}, & m_i^{j} = n_i^{j} - 1, \\
	\begin{pmatrix}
		0 & 0 \\
		0 & 0
	\end{pmatrix}, & \text{otherwise}.
\end{cases}
\end{align}
$|m ^j\rangle^{(j)} \not\triangleright |n^j \rangle^{(j)} $ implies that at least two matrices $A_{12}$ or $A_{21}$ will appear consecutively, which leads to
\begin{align}
\uparrow^*
\bar{L}_{M_1}^{(1)}(x)\cdots \bar{L}_0^{(1)}(x)\downarrow \uparrow^* \bar{L}_{M_2}^{(2)}(y)\cdots \bar{L}_0^{(2)}(y)
\downarrow=0.
\end{align}
In the case of $|m ^j\rangle^{(j)} \triangleright |n^j \rangle^{(j)} $,
let all integers $i$ satisfying $m^{j}_{i} = n_i^{j} + 1$ form the set $\{p_1^{j},\cdots,p_{r_j}^{j} \}$ and let $\{q_1^{j},\cdots,q_{s_j}^{j} \}$ formed by integers $i$ satisfying $m_i^{j} = n_i^{j} - 1$.
From  the admissibility relation, one gets $r_j -1 = s_j$ and
\begin{align}
p_i^{j} < q_i^{j} < p_{i+1}^{j}, \quad \text{for all } 1 \leq i \leq r_j - 1.
\end{align}
Note that $q_0^{j}=-1$ and $q_{r_j}^{j}=M_j +1$. The Eq.$(\ref{B1-B2})$ was rewritten as
\begin{align}
&\uparrow^*
\bar{L}_{M_1}^{(1)}(x)\cdots \bar{L}_0^{(1)}(x)A_{21}\uparrow
\uparrow^* \bar{L}_{M_2}^{(2)}(y)\cdots \bar{L}_0^{(2)}(y) A_{21}
\uparrow    \notag\\
=&\uparrow^* \rprod_{i=1}^{r_1}
\left[
\begin{pmatrix}
	x^{-\frac{1}{2}} & 0 \\
	0 & x^{\frac{1}{2}}
\end{pmatrix}^{q_i^1 - p_{i-1}^1 - 1}
\begin{pmatrix}
	0 & 1 \\
	0 & 0
\end{pmatrix}
\begin{pmatrix}
	x^{-\frac{1}{2}} & 0 \\
	0 & x^{\frac{1}{2}}
\end{pmatrix}^{p_i^1 - q_{i-1}^1 - 1}
\begin{pmatrix}
	0 & 0 \\
	1 & 0
\end{pmatrix}
\right] \uparrow  \notag\\
&\uparrow^* \rprod_{l=1}^{r_2}
\left[
\begin{pmatrix}
	y^{-\frac{1}{2}} & 0 \\
	0 & y^{\frac{1}{2}}
\end{pmatrix}^{q_l^2 - p_{l-1}^2 - 1}
\begin{pmatrix}
	0 & 1 \\
	0 & 0
\end{pmatrix}
\begin{pmatrix}
	y^{-\frac{1}{2}} & 0 \\
	0 & y^{\frac{1}{2}}
\end{pmatrix}^{p_l^2 - q_{l-1}^2 - 1}
\begin{pmatrix}
	0 & 0 \\
	1 & 0
\end{pmatrix}
\right] \uparrow  \notag\\
=&x^{-\frac{M_1}{2}} \prod_{i=1}^{r_1 -1} x^{p_i^1 - q_i^1} x^{p^1_{r_1}}
\cdot
y^{-\frac{M_2}{2}} \prod_{l=1}^{r_2 -1} y^{p_l^2 - q_l^2} y^{p^2_{r_2}}
\notag\\
=& x^{-\frac{M_1}{2}} \prod_{i=1}^{M_1} x^{i(m_i^1 - n_i^1)}
\cdot
y^{-\frac{M_2}{2}} \prod_{l=1}^{M_2} y^{l(m_l^2 - n_l^2)}.
\end{align}
Therefore,
\begin{align}
	&x^{\frac{M_1}{2}}y^{\frac{M_2}{2}}\prescript{(2)}{}{\langle m^{2}|}\bigotimes \prescript{(1)}{}{\langle m^{1}|}
	B_1 (x) B_2 (y)
	|n^{1}\rangle^{(1)}\bigotimes|n^{2}\rangle^{(2)}\notag\\
	=&\prescript{(2)}{}{\langle m^{2}|}\bigotimes \prescript{(1)}{}{\langle m^{1}|}
\bar{\mathbb{B}}_{1}(x)\bar{\mathbb{B}}_{2}(y)
	|n^{1}\rangle^{(1)}\bigotimes|n^{2}\rangle^{(2)}\notag\\
=&
\begin{cases}
\prod\limits_{i=1}^{M_1} x^{i(m_i^1 - n_i^1)}\prod\limits_{l=1}^{M_2} y^{l(m_l^2 - n_l^2)},
& | m^{1}\rangle^{(1)} \triangleright   | n^{1}\rangle^{(1)} \text{and }
| m^{2}\rangle^{(2)} \triangleright   |   n^{2}\rangle^{(2)}, \\
	0, & \text{otherwise}.
\end{cases}
\end{align}
Based on the Eq.$(\ref{tgpm-scalar-product})$, the Eq.$(\ref{1-B-n,C-n,2})$ can be derived.
\end{proof}
The above derivation shows that the action of $\bar{\mathbb{B}}_{1}(x)$ and $\bar{\mathbb{B}}_{2}(x)$ on $|n^{1}\rangle^{(1)}\bigotimes|n^{2}\rangle^{(2)} $ is independent of each other, and the same holds true for $\bar{\mathbb{C}}_{1}(x)$ and $\bar{\mathbb{C}}_{2}(x)$.
It implies that
\begin{align}
\bar{\mathbb{B}}_{1}(x)\bar{\mathbb{B}}_{2}(y)|n^{1}\rangle^{(1)}\bigotimes|n^{2}\rangle^{(2)}
=&\bar{\mathbb{B}}_{1}(x)|n^{1}\rangle^{(1)}  \bigotimes\bar{\mathbb{B}}_{2}(y)|n^{2}\rangle^{(2)}, \label{B1B2-B1-B2}\\
\prescript{(2)}{}{\langle n^{2}|}\bigotimes \prescript{(1)}{}{\langle n^{1}|} \bar{\mathbb{C}}_{1}(x)\bar{\mathbb{C}}_{2}(x)
=&\prescript{(2)}{}{\langle n^{2}|}\bar{\mathbb{C}}_{2}(x)  \bigotimes
\prescript{(1)}{}{\langle n^{1}|}\bar{\mathbb{C}}_{1}(x),
\end{align}
and
\begin{align}
\bar{\mathbb{B}}_{1}(x)|n^{1}\rangle^{(1)}&=\sum_{| m^{1} \rangle^{(1)} \triangleright | n^{1} \rangle^{(1)}}\prod_{i=1}^{M_1} x^{i(m^{1}_i - n^{1}_i)} | m^{1} \rangle ^{(1)}, \label{B1-1}
\\
\bar{\mathbb{B}}_{2}(y)|n^{2}\rangle^{(2)}&=\sum_{| m^{2} \rangle^{(2)} \triangleright | n^{2} \rangle^{(2)}}\prod_{j=1}^{M_2} y^{j(m^{2}_j - n^{2}_j)} | m^{2} \rangle^{(2)},\label{B2-1}\\
\prescript{(1)}{}{\langle n^{1}|}\bar{\mathbb{C}}_{1}(x)
&=\sum_{\prescript{(1)}{}{\langle n^{1}|}  \triangleleft \prescript{(1)}{}{\langle m^{1}|}}\prod_{i=1}^{M_1} x^{i(m^{1}_i - n^{1}_i)}\prescript{(1)}{}{\langle m^{1}|},\label{C1-1}\\
\prescript{(2)}{}{\langle n^{2}|}\bar{\mathbb{C}}_{2}(x)
&=\sum_{\prescript{(2)}{}{\langle n^{2}|} \triangleleft \prescript{(2)}{}{\langle m^{2}|}} \prod_{j=1}^{M_2} y^{j(m^{2}_j - n^{2}_j)}  \prescript{(2)}{}{\langle m^{2}|} \label{C2-1}.
\end{align}
\begin{proposition}\label{M-B-C-2}
	Set $l_j=l(\nu_j)(j=1,2)$, from the relations
	\begin{align}
		\mathcal{M}  | n^{1}\rangle^{(1)}\bigotimes|n^{2}\rangle^{(2)}&=| \nu^{1},\nu^{2} \rangle ,
		\quad 	\mathcal{M}  | m^{1}\rangle^{(1)}\bigotimes|m^{2}\rangle^{(2)}=| \mu^{1},\mu^{2} \rangle,
		\\
		\prescript{(2)}{}{\langle n^{2}|}\bigotimes \prescript{(1)}{}{\langle n^{1}|}  \mathcal{M}^* &= \langle \nu^{2},\nu^{1}|, \quad
		\prescript{(2)}{}{\langle m^{2}|}\bigotimes \prescript{(1)}{}{\langle m^{1}|}  \mathcal{M}^* = \langle \mu^{2},\mu^{1}| ,
	\end{align}
	we have
	\begin{align}
		\mathcal{M} \bar{\mathbb{B}}_{1}(x)\bar{\mathbb{B}}_{2}(y)|n^{1}\rangle^{(1)}\bigotimes|n^{2}\rangle^{(2)}
		&= \sum_{{ \nu^{1} \prec \mu^{1} \subseteq [l_{1}+1,M_{1}] }\atop{\nu^{2} \prec \mu^{2} \subseteq [l_{2}+1,M_{2}]}}
		x^{|\mu^{1}| - |\nu^{1}|} y^{|\mu^{2}| - |\nu^{2}|}
		| \mu^{1},\mu^{2} \rangle	,\label{M-B1-B2-2}	
		\\
		\prescript{(2)}{}{\langle n^{2}|}\bigotimes \prescript{(1)}{}{\langle n^{1}|}
		\bar{\mathbb{C}}_{2}(y)\bar{\mathbb{C}}_{1}(x)
		\mathcal{M}^*
		&=
		\sum_{{ \nu^{1} \prec \mu^{1} \subseteq [l_{1}+1,M_{1}] }\atop{\nu^{2} \prec \mu^{2} \subseteq [l_{2}+1,M_{2}]}}
		x^{|\mu^{1}| - |\nu^{1}|} y^{|\mu^{2}| - |\nu^{2}|}
		\langle \mu^{2},\mu^{1} |,\label{M-C1-C2-2}	
	\end{align}	
	which are summed over  boxed partitions. Here $[l_{j}+1, M_j](j=1,2)$ denote a rectangular Young diagrams with $l_{j}+1$ rows and $M_j$ columns.
\end{proposition}
\begin{proof}
Since $l_j=l(\nu_j)$ and $\nu^{j} \prec \mu^{j}(j=1,2)$, then $l_j +1=l(\mu_j)$ and $\mu^{j} \subseteq [l_{j}+1,M_{j}]$.
Based on the Definition $\ref{M}$,  we have $|\mu^{j}|=\sum\limits_{i=1}^{M_j}m^{j}_i \cdot i$ and $|\nu^{j}|=\sum\limits_{i=1}^{M_j}n^{j}_i \cdot i$.
Combining Definition $\ref{M}$, Lemma \ref{xr-jc-1} and Lemma \ref{B-n,C-n,2} yields Eqs.(\ref{M-B1-B2-2}) and (\ref{M-C1-C2-2}).
\end{proof}
Set $\{\mathbf{x}\}=\{x_1,\cdots,x_{n}\}$. Introduce the skew Schur function
\begin{align}\label{skew-schur}
s_{\mu/\nu}(x) = \begin{cases}
	x^{|\mu| - |\nu|}, & \mu \succ \nu, \\
	0, & \text{otherwise},
\end{cases}
\end{align}
where the Schur function is obtained by setting $\nu=\emptyset$ and has the following form
\begin{align}\label{schur}
s_\mu\{\mathbf{x}\}= \frac{\det \left( x_i^{\mu_j - j + n} \right)_{1 \leq i,j \leq n}}{\prod\limits_{1 \leq i < j \leq n} (x_i - x_j)},
\end{align}
which satisfies
\begin{align}\label{schur-skew-schur}
	s_\mu\{x_1, \ldots, x_n\} = \sum_{\nu \subseteq [n-1,\infty]} s_{\mu/\nu}(x_n) \, s_\nu\{x_1, \ldots, x_{n-1}\}.	
\end{align}	
Combining the Proposition \ref{M-B-C-2} and the Eq.$(\ref{skew-schur})$, one gets following properties.
\begin{proposition}
The following equations hold
\begin{align}
	\mathcal{M} \bar{\mathbb{B}}_{1}(x)\bar{\mathbb{B}}_{2}(y)|n^{1}\rangle^{(1)}\bigotimes|n^{2}\rangle^{(2)}
	&= \sum_{{ \nu^{1} \prec \mu^{1} \subseteq [l_{1}+1,M_{1}] }\atop{\nu^{2} \prec \mu^{2} \subseteq [l_{2}+1,M_{2}]}}
	s_{\mu^{1}/\nu^{1}}(x) s_{\mu^{2}/\nu^{2}}(y)
	| \mu^{1},\mu^{2} \rangle,	\label{M-B1-B2-N}
	\\
	\prescript{(2)}{}{\langle n^{2}|}\bigotimes \prescript{(1)}{}{\langle n^{1}|}
	\bar{\mathbb{C}}_{2}(y)\bar{\mathbb{C}}_{1}(x)
	\mathcal{M}^*
	&=	\sum_{{ \nu^{1} \prec \mu^{1} \subseteq [l_{1}+1,M_{1}] }\atop{\nu^{2} \prec \mu^{2} \subseteq [l_{2}+1,M_{2}]}}
	s_{\mu^{1}/\nu^{1}}(x) s_{\mu^{2}/\nu^{2}}(y)
	\langle \mu^{2},\mu^{1} |.\label{M-C1-C2-N}
\end{align}	

\end{proposition}

\section{The scalar product and boxed UC plane partitions}
In this section, we construct the boxed UC plane partition and the scalar product of the  two-site generalized phase model.
With the help of actions of a series of  monodromy matrix operators on the basis vectors, we derive the generating function of boxed UC plane partitions.
Meanwhile, in term of the skew Schur functions, the generating function can be expressed as products  of Schur functions.

\begin{definition}
	For $i \in \mathbb{Z}$ and a pair of  boxed partitions $\alpha^i$ and $\beta^i$, denote the boxed $2-$partition $(\chi_{i})=(\alpha^i,\beta^i)$.
	The boxed UC plane partition refers to the UC plane partition placed into a box of a finite size, which  is defined as $\Pi^{'}=(\ldots,\chi_{-1},\chi_{0},\chi_{1},\ldots)$
	and represents a pair of boxed plane partitions $\pi^{1}=(\ldots,\alpha^{-1},\alpha^{0},\alpha^{1},\ldots)$ and
$\pi^{2}=(\ldots,\beta^{-1},\beta^{0},\beta^{1},\ldots)$.	
\end{definition}
\begin{proposition}\label{M-n-0}
The following equations hold
\begin{align}
	\mathcal{M}
\lprod_{i=1}^{N_1} \bar{\mathbb{B}}_{1}(x_i)
\lprod_{j=1}^{N_2} \bar{\mathbb{B}}_{2}(y_j)
	|0\rangle^{(1)}\bigotimes|0\rangle^{(2)}
	&= \sum_{{ [N_{1},M_{1}] \ni \pi^{1}_0 \succ \cdots \succ \pi^{1}_{N_{1}} = \emptyset}
		\atop{[N_{2},M_{2}] \ni \pi^{2}_0 \succ \cdots \succ \pi^{2}_{N_{2}} = \emptyset}} 	
	\prod_{i=1}^{N_1} x_i^{|\pi^{1}_{i-1}| - |\pi^{1}_i|}
	\prod_{j=1}^{N_2} y_j^{|\pi^{2}_{j-1}| - |\pi^{2}_j|}
	| \pi^{1}_0,\pi^{2}_0 \rangle,\label{M-B1-BN-x}\\
	\prescript{(2)}{}{\langle 0|}\bigotimes \prescript{(1)}{}{\langle 0|}
\rprod_{j=1}^{N_2} \bar{\mathbb{C}}_{2}(y_j)
\rprod_{i=1}^{N_1} \bar{\mathbb{C}}_{1}(x_i) \mathcal{M}^*
	&=
	\sum_{{ \emptyset=\pi^{1}_{-N_{1}} \prec \cdots \prec \pi^{1}_0  \in [N_{1}, M_{1}]}
		\atop{\emptyset=\pi^{2}_{-N_{2}} \prec \cdots \prec \pi^{2}_0  \in [N_{2}, M_{2}]}}
	\prod_{i=1}^{N_1} x_i^{|\pi^{1}_{i-1}| - |\pi^{1}_i|}
\prod_{j=1}^{N_2} y_j^{|\pi^{2}_{j-1}| - |\pi^{2}_j|}
	\langle \pi^{2}_0,\pi^{1}_0|,\label{M-C1-CN-x}
\end{align}	
which are summed over $[N_{j},M_{j}] \ni \pi^{j}_0 \succ \cdots \succ \pi^{j}_{N_{j}} = \emptyset$ and $\emptyset=\pi^{j}_{-N_{j}} \prec \cdots \prec \pi^{j}_0  \in [N_{j}, M_{j}]$ respectively.
\end{proposition}
\begin{proof}
The Eq.$(\ref{B1B2-B1-B2})$ leads to
\begin{align}
\lprod_{i=1}^{N_1} \bar{\mathbb{B}}_{1}(x_i)
\lprod_{j=1}^{N_2} \bar{\mathbb{B}}_{2}(y_j)
|0\rangle^{(1)}\bigotimes|0\rangle^{(2)}
=\lprod_{i=1}^{N_1} \bar{\mathbb{B}}_{1}(x_i) |0\rangle^{(1)}
\bigotimes
\lprod_{j=1}^{N_2} \bar{\mathbb{B}}_{2}(y_j)|0\rangle^{(2)}.
\end{align}
By repeatedly iterating Eqs.$(\ref{B1-1})$ and $(\ref{B2-1})$  $N_1$ and $N_2$ times with $|n^{1}\rangle^{(1)}\bigotimes|n^{2}\rangle^{(2)}=|0\rangle^{(1)}\bigotimes|0\rangle^{(2)}$, respectively, we obtain
\begin{align}
	&\lprod_{i=1}^{N_1} \bar{\mathbb{B}}_{1}(x_i) |0\rangle^{(1)}
	\bigotimes
	\lprod_{j=1}^{N_2} \bar{\mathbb{B}}_{2}(y_j)|0\rangle^{(2)}\notag\\
=&\sum_{{|m^{a}\rangle^{(1)} \triangleright |m^{a+1}\rangle^{(1)}}
	\atop{a=1,\cdots,N_1}}
\sum_{{|n^{b}\rangle^{(2)} \triangleright |n^{b+1}\rangle^{(2)}}
	\atop{b=1,\cdots,N_2}}
	\prod_{l=1}^{N_1}\prod_{i=1}^{M_1} x^{i(m^{l}_i - m^{1+1}_i)}
	\prod_{k=1}^{N_2}\prod_{j=1}^{M_2} y^{j(n^{k}_j - n^{k+1}_j)}
	|m^{1}\rangle^{(1)}\bigotimes|n^{1}\rangle^{(2)},
\end{align}
where $|m^{N_1+1}\rangle^{(1)}=|0\rangle^{(1)}$ and  $|n^{N_2+1}\rangle^{(2)}=|0\rangle^{(2)}$.

From the assumption
\begin{align}
\mathcal{M} |m^{a}\rangle^{(1)}\bigotimes|n^{b}\rangle^{(2)}
=|\pi^{a-1}\rangle^{(1)}\bigotimes|\pi^{b-1}\rangle^{(2)},
\end{align}
one gets
\begin{align}
	&\mathcal{M}
	\lprod_{i=1}^{N_1} \bar{\mathbb{B}}_{1}(x_i)
	\lprod_{j=1}^{N_2} \bar{\mathbb{B}}_{2}(y_j)
	|0\rangle^{(1)}\bigotimes|0\rangle^{(2)}\notag\\
	= &\sum_{{ [N_{1},M_{1}] \ni \pi^{1}_0 \succ \cdots \succ \pi^{1}_{N_{1}} = \emptyset}
		\atop{[N_{2},M_{2}] \ni \pi^{2}_0 \succ \cdots \succ \pi^{2}_{N_{2}} = \emptyset}} 	
	\prod_{i=1}^{N_1} x_i^{|\pi^{1}_{i-1}| - |\pi^{1}_i|}
	\prod_{j=1}^{N_2} y_j^{|\pi^{2}_{j-1}| - |\pi^{2}_j|}
	| \pi^{1}_0,\pi^{2}_0 \rangle.
\end{align}	
Using the similar approach, the Eq.$(\ref{M-C1-CN-x})$ can be proved.
\end{proof}

Let $\{\mathbf{a^j}\}=\{a_1,\cdots,a_{N_j}\}$, where $\mathbf{a}$  can be replaced by $\mathbf{x},\mathbf{y},\mathbf{z},\mathbf{v}$ and $j=1,2$.
Consider the scalar product of the two-site generalized phase model
\begin{align}
	&S(\{\mathbf{x^1}\}, \{\mathbf{y^2}\}, \{\mathbf{z^1}\}, \{\mathbf{v^2}\})\notag\\
	=&\bigg \langle \prescript{(2)}{}{\langle 0|}\bigotimes \prescript{(1)}{}{\langle 0|}
	\rprod_{j=1}^{N_2} \bar{\mathbb{C}}_{2}(y_j)
	\rprod_{i=1}^{N_1} \bar{\mathbb{C}}_{1}(x_i) \mathcal{M}^*
	,\mathcal{M}\lprod_{l=1}^{N_1} \bar{\mathbb{B}}_{1}(z_i)
	\lprod_{k=1}^{N_2} \bar{\mathbb{B}}_{2}(v_k)
	|0\rangle^{(1)}\bigotimes|0\rangle^{(2)} \bigg\rangle.
\end{align}
From the Eq.$(\ref{scalar-product-equal})$, one obtains
\begin{align}
	&S(\{\mathbf{x^1}\}, \{\mathbf{y^2}\}, \{\mathbf{z^1}\}, \{\mathbf{v^2}\})\notag\\
	=&\prescript{(2)}{}{\langle 0|}\bigotimes \prescript{(1)}{}{\langle 0|}
	\rprod_{j=1}^{N_2} \bar{\mathbb{C}}_{2}(y_j)
	\rprod_{i=1}^{N_1} \bar{\mathbb{C}}_{1}(x_i)
	\lprod_{l=1}^{N_1} \bar{\mathbb{B}}_{1}(z_i)
	\lprod_{k=1}^{N_2} \bar{\mathbb{B}}_{2}(v_k)
	|0\rangle^{(1)}\bigotimes|0\rangle^{(2)}.
\end{align}
By using the Lemma $\ref{M-n-0}$ to yield
\begin{align}\label{scalar-product-A1-A2}
S(\{\mathbf{x^1}\}, \{\mathbf{y^2}\}, \{\mathbf{z^1}\}, \{\mathbf{v^2}\})
 =&
	  \sum_{{\pi^{1} \subseteq [N_{1}, N_{1}, M_{1}]}
	  	\atop{\pi^{2} \subseteq [N_{2}, N_{2}, M_{2}]}}
A_{\pi^{1} }(\{\mathbf{x^1}\}, \{\mathbf{z^1}\}) A_{\pi^{2} }(\{\mathbf{y^2}\}, \{\mathbf{v^2}\}),
\end{align}
where we note  $\Pi^{'}=(\ldots,\chi_{-1},\chi_{0},\chi_{1},\ldots)$,
$(\chi_{i})=(\pi_{i}^{1},\pi_{i}^{2})$,
	$\pi^{t}=(\ldots,\pi_{-1}^{t},\pi_{0}^{t},\pi_{1}^{t},\ldots)(t=1,2)$ and
\begin{align}
A_{\pi^{1} }(\{\mathbf{x^1}\}, \{\mathbf{z^1}\})
&=\prod_{i,l=1}^{N_1}
x_i^{|\pi^{1}_{i-1}| - |\pi^{1}_i|}
z_l^{|\pi^{2}_{l-1}| - |\pi^{2}_l|},\\
A_{\pi^{2} }(\{\mathbf{y^2}\}, \{\mathbf{v^2}\})
&=
\prod_{j,k=1}^{N_2}
y_j^{|\pi^{1}_{j-1}| - |\pi^{1}_j|}
v_k^{|\pi^{2}_{k-1}| - |\pi^{2}_k|}	.
\end{align}
It is found from the Eq.($\ref{scalar-product-A1-A2}$)   that the scalar product   is a generating function of  boxed UC plane partitions.

\begin{proposition}\label{M-n-0-schur}
	The following equations hold
	\begin{align}
		\mathcal{M}
		\lprod_{i=1}^{N_1} \bar{\mathbb{B}}_{1}(x_i)
		\lprod_{j=1}^{N_2} \bar{\mathbb{B}}_{2}(y_j)
		|0\rangle^{(1)}\bigotimes|0\rangle^{(2)}
		= &\sum_{{ \mu^{1} \subseteq [N_1,M_1]}\atop{\mu^{2} \subseteq [N_2,M_2]}}
		s_{\mu^{1}}\{\mathbf{x^1}\}  s_{\mu^{2}}\{\mathbf{y^2}\}
		| \mu^{1},\mu^{2} \rangle,\label{M-B1-BN-schur}	\\	
		\prescript{(2)}{}{\langle 0|}\bigotimes \prescript{(1)}{}{\langle 0|}
		\rprod_{j=1}^{N_2} \bar{\mathbb{C}}_{2}(y_j)
		\rprod_{i=1}^{N_1} \bar{\mathbb{C}}_{1}(x_i)\mathcal{M}^*
		= &\sum_{{ \mu^{1} \subseteq [N_1,M_1]}\atop{\mu^{2} \subseteq [N_2,M_2]}}
		s_{\mu^{1}}\{\mathbf{x^1}\}  s_{\mu^{2}}\{\mathbf{y^2}\}
		\langle \mu^{2},\mu^{1}|.\label{M-C1-CN-schur}
	\end{align}	
\end{proposition}
\begin{proof}
	Considering the special case of $|n^{1}\rangle^{(1)}\bigotimes|n^{2}\rangle^{(2)}=|0\rangle^{(1)}\bigotimes|0\rangle^{(2)}$ in the Eq.$(\ref{M-C1-C2-N})$ yields
	\begin{align}
		\mathcal{M} \bar{\mathbb{B}}_{1}(x)\bar{\mathbb{B}}_{2}(y) |0\rangle^{(1)}\bigotimes|0\rangle^{(2)}
		=& \sum_{{ \emptyset \prec \mu^{1} \subseteq [1,M_{1}]}\atop{\emptyset \prec \mu^{2} \subseteq [1,M_{2}]}}
		s_{{\mu^{1}}/\emptyset}(x) s_{{\mu^{2}}/\emptyset}(y)
		|\mu^{1},\mu^{2}\rangle \notag\\
		=&\sum_{{ \emptyset \prec \mu^{1} \subseteq [1,M_{1}]}\atop{\emptyset \prec \mu^{2} \subseteq [1,M_{2}]}}
		s_{\mu^{1}}(x) s_{\mu^{2}}(y)  |\mu^{1},\mu^{2}\rangle .
	\end{align}
	Meanwhile, for all $N_{1}, N_{2}\geq 2$, assume that
	\begin{align}\label{M-N-1}
		\mathcal{M}
		\prod_{i=1}^{N_{1}-1} \bar{\mathbb{B}}_{1}(x_i)
		\prod_{j=1}^{N_{2}-1} \bar{\mathbb{B}}_{2}(y_i)
		|0\rangle^{(1)}\bigotimes|0\rangle^{(2)}
		= &\sum_{{ \nu^{1} \subseteq [N_{1}-1,M_1]}\atop{\nu^{2} \subseteq [N_{2}-1,M_2]}}
		s_{\nu^{1}}\{\bar{\mathbf{x}}^1\}  s_{\nu^{2}}\{ \bar{\mathbf{y}}^2 \}
		| \nu^{1},\nu^{2} \rangle,
	\end{align}	
where
\begin{align}
\{\bar{\mathbf{x}}^1\}=\{x_1,\cdots,x_{N_1-1}\},\quad  \{\bar{\mathbf{y}}^2\}=\{y_1,\cdots,y_{N_2-1}\}.
\end{align}
Combining the definition of the map $\mathcal{M}$ and the Eq.(\ref{M-N-1}) yields
	\begin{align}
		\prod_{i=1}^{N_{1}-1} \bar{\mathbb{B}}_{1}(x_i)
		\prod_{j=1}^{N_{2}-1} \bar{\mathbb{B}}_{2}(y_i)
		|0\rangle^{(1)}\bigotimes|0\rangle^{(2)}
		= \sum_{{ 	| n^{1}\rangle^{(1)}|\Sigma_0^{n^{1}} =N_{1}-1}\atop{| n^{2}\rangle^{(2)}|\Sigma_0^{n^{2}}=N_{2}-1} }
		s_{\nu^{1}}\{\bar{\mathbf{x}}^1\}  s_{\nu^{2}}\{ \bar{\mathbf{y}}^2 \}
		| n^{1}\rangle^{(1)}\bigotimes|n^{2}\rangle^{(2)},
	\end{align}	
	where the sum is  over all basis vectors $| n^{1}\rangle^{(1)}\bigotimes|n^{2}\rangle^{(2)}$, which satisfies occupation numbers $\Sigma_0^{n^j} = \sum\limits_{k=0}^{M_j} n_k^j=N_{j}-1$ $(j=1,2)$  and corresponds to the $2-$partition $(\chi)=(\nu^{1},\nu^{2})$.
	Using the Eq.$(\ref{B,C,commute})$ and acting $\mathcal{M} \bar{\mathbb{B}}_{1}(x_{N_1})\bar{\mathbb{B}}_{2}(y_{N_2})$ on the above equation to obtain
	\begin{align}
		&\mathcal{M} \bar{\mathbb{B}}_{1}(x_{N_1})\bar{\mathbb{B}}_{2}(y_{N_2})
		\prod_{i=1}^{N_{1}-1} \bar{\mathbb{B}}_{1}(x_i)
		\prod_{j=1}^{N_{2}-1} \bar{\mathbb{B}}_{2}(y_i)	
		|0\rangle^{(1)}\bigotimes|0\rangle^{(2)}\notag\\
		=&\sum_{{ \nu^{1} \subseteq [N_{1}-1,M_1]}\atop{\nu^{2} \subseteq [N_{2}-1,M_2]}}
		s_{\nu^{1}}\{\bar{\mathbf{x}}^1\}  s_{\nu^{2}}\{ \bar{\mathbf{y}}^2 \} 	
		\sum_{{\nu^{1} \prec \mu^{1} \subseteq [N_1,M_1]}\atop{\nu^{2} \prec \mu^{2} \subseteq [N_2,M_2]}}
		s_{\mu^{1}/\nu^{1}}(x_{N_1}) s_{\mu^{2}/\nu^{2}}(y_{N_2}) 	
		| \mu^{1},\mu^{2} \rangle.
	\end{align}	
	Consider the interlacing relation $\nu^{j} \prec \mu^{j}$ and the Eq.$(\ref{schur-skew-schur})$, we get
	\begin{align}
		&\mathcal{M}
		\lprod_{i=1}^{N_1} \bar{\mathbb{B}}_{1}(x_i)
		\lprod_{j=1}^{N_2} \bar{\mathbb{B}}_{2}(y_j)
		|0\rangle^{(1)}\bigotimes|0\rangle^{(2)}\notag\\
		= &	\sum_{{ \mu^{1} \subseteq [N_1,M_1]}\atop{\mu^{2} \subseteq [N_2,M_2]}}
		\sum_{{ \nu^{1} \subseteq [N_{1}-1,\infty]}\atop{\nu^{2} \subseteq [N_{2}-1,\infty]}}
		s_{\mu^{1}/\nu^{1}}(x_{N_1}) 	s_{\nu^{1}}\{\bar{\mathbf{x}}^1\}
		s_{\mu^{2}/\nu^{2}}(y_{N_2}) 	s_{\nu^{2}}\{ \bar{\mathbf{y}}^2 \}
		| \mu^{1},\mu^{2} \rangle \notag\\	
		= &\sum_{{ \mu^{1} \subseteq [N_1,M_1]}\atop{\mu^{2} \subseteq [N_2,M_2]}}
		s_{\mu^{1}}\{\mathbf{x^1}\}  s_{\mu^{2}}\{\mathbf{y^2}\}
		| \mu^{1},\mu^{2} \rangle,	
	\end{align}	
	which means that it holds for any $N_{1}, N_{2} \geq 1$.

	The proof of the Eq.(\ref{M-C1-CN-schur}) is quite similar, so it is omitted.
\end{proof}
From the Lemma $\ref{M-n-0-schur}$, the scalar product can be rewritten as
\begin{align}\label{scalar-product-S1-S2}
&S(\{\mathbf{x^1}\}, \{\mathbf{y^2}\}, \{\mathbf{z^1}\}, \{\mathbf{v^2}\})\notag\\	
 =&\prescript{(2)}{}{\langle 0|}\bigotimes \prescript{(1)}{}{\langle 0|}
\rprod_{j=1}^{N_2} \bar{\mathbb{C}}_{2}(y_j)
\rprod_{i=1}^{N_1} \bar{\mathbb{C}}_{1}(x_i)
\lprod_{l=1}^{N_1} \bar{\mathbb{B}}_{1}(z_i)
\lprod_{k=1}^{N_2} \bar{\mathbb{B}}_{2}(v_j)
|0\rangle^{(1)}\bigotimes|0\rangle^{(2)}\notag\\	
=&\sum_{{ \mu^{1} \subseteq [N_1,M_1]}\atop{\mu^{2} \subseteq [N_2,M_2]}}
s_{\mu^{1}}\{\mathbf{x^1}\}  s_{\mu^{1}}\{\mathbf{z^1}\}
s_{\mu^{2}}\{\mathbf{y^2}\}  s_{\mu^{2}}\{\mathbf{v^2}\},
\end{align}	
which implies
\begin{align}
\sum_{{\pi^{1} \subseteq [N_{1}, N_{1}, M_{1}]}
	\atop{\pi^{2} \subseteq [N_{2}, N_{2}, M_{2}]}}
A_{\pi^{1} }(\{\mathbf{x^1}\}, \{\mathbf{z^1}\}) A_{\pi^{2} }(\{\mathbf{y^2}\}, \{\mathbf{v^2}\})
=
\sum_{{ \mu^{1} \subseteq [N_1,M_1]}\atop{\mu^{2} \subseteq [N_2,M_2]}}
s_{\mu^{1}}\{\mathbf{x^1}\}  s_{\mu^{1}}\{\mathbf{z^1}\}
s_{\mu^{2}}\{\mathbf{y^2}\}  s_{\mu^{2}}\{\mathbf{v^2}\}.
\end{align}
The generating function of  boxed UC plane partitions has been represented as  products  of Schur functions.
%
\section{The generating function of boxed UC plane partitions with the double scaling limit}
This section investigates  interlacing $2-$partitions  in terms of the charged fermionic Fock space. The two-side generalized phase model on infinite lattice limit ($M_{1},M_{2} \to \infty$) is studied. We prove that  the generating function of boxed UC plane partitions leads to the generating function of UC plane partitions with the double scaling limit, which means that the number of lattice sites $M_{1},M_{2} \to \infty$ and the number of particles $N_{1},N_{2} \to \infty$.
\subsection{Generating interlacing $2-$partitions}
Define vertex operators
\begin{align}\label{UCH}
	\Gamma_+(z,v) &= e^{H_{+}(z,v)}=
	\exp\left( \sum_{n=1}^{\infty} \left(\frac{z^n}{n} H_{n} + \frac{v^n}{n} \tilde{H}_{n} \right)\right), \\
	\Gamma_-(z,v) &=  e^{H_{-}(z,v)}=
	\exp\left( \sum_{n=1}^{\infty}  \left(\frac{z^n}{n} H_{-n} + \frac{v^n}{n} \tilde{H}_{-n} \right) \right).
\end{align}
Since
\begin{align}
	[H_{+}(z,v),H_{-}(x,y)]
	&=\sum_{m,n=1}^{\infty}   \frac{z^{m}x^{n}}{mn}[H_{m},H_{-n}]
	+\sum_{m,n=1}^{\infty}   \frac{v^{m}y^{n}}{mn}[\tilde{H}_{m},\tilde{H}_{-n}] \notag\\
	&=\sum_{m=1}^{\infty} \frac{(zx)^{m}}{m} +  \sum_{n=1}^{\infty} \frac{(vy)^{n}}{n},
\end{align}
then we have
\begin{align}\label{Gamme-commutation-relations}
	\Gamma_+(z,v) \Gamma_-(x,y)&=e^{[H_{+}(z,v),H_{-}(x,y)]} \Gamma_-(x,y)\Gamma_+(z,v)\notag\\
	&=\frac{1}{1-zx} \frac{1}{1-vy}   \Gamma_-(x,y)\Gamma_+(z,v).
\end{align}

\begin{proposition}\label{H1}
	For generating functions of charged fermions
	\begin{align}\label{UC-feimisc}
		\psi(k)&=\sum_{n\in\mathbb{Z}+1/2}\psi_nk^{-n-1/2},\quad \psi^*(k)=\sum_{n\in\mathbb{Z}+1/2}\psi_n^*k^{-n-1/2},\\
		\phi(k)&=\sum_{n\in\mathbb{Z}+1/2}\phi_nk^{-n-1/2},\quad \phi^*(k)=\sum_{n\in\mathbb{Z}+1/2}\phi_n^*k^{-n-1/2},
	\end{align}
	the following equations hold
	\begin{align}\label{commutation relations 3.1}
		[H_{+}(z,v),\psi(k)]&=\sum_{n=1}^{\infty} \frac{(zk)^n}{n} \psi(k),\quad
		[H_{-}(z,v), \psi^{\ast}(k)]
		=-\sum_{n=1}^{\infty} \frac{1}{n}  \left(\frac{z}{k}\right)^n \psi^{\ast}(k),\\
		[H_{+}(z,v),\phi(k)]&=\sum_{n=1}^{\infty} \frac{(vk)^n}{n} \phi(k),\quad
		[H_{-}(z,v),\phi^{\ast}(k)]=-\sum_{n=1}^{\infty} \frac{1}{n}   \left(\frac{v}{k}\right)^n \phi^{\ast}(k).
	\end{align}
\end{proposition}
\begin{proof}
	Combining Eqs.(\ref{commutation relations 2.1}) and (\ref{UCH}) yields
	\begin{align}
		[H_{+}(z,v),\psi(k)]
		&=\sum_{n=1}^{\infty}\sum_{m\in\mathbb{Z}+1/2}
		\frac{z^n}{n}\Big[H_{n},\psi_m\Big]
		\cdot k^{-m-\frac{1}{2}} \notag\\
		&=\sum_{n=1}^{\infty}   \frac{z^n}{n}k^n
		\sum_{m^{'}\in\mathbb{Z}+1/2}\psi_{m^{'}} k^{-m^{'}-1/2}
		\notag\\
		&= \sum_{n=1}^{\infty} \frac{(zk)^n}{n}   \psi(k).
	\end{align}
	Other formulas can be proved by the same method.	
\end{proof}
\begin{lemma}\label{Gamma-(k)}
	The following equations hold
	\begin{align}\label{commutation relations 4.1}	
	&	\Gamma_+(z,v)\psi(k)\Gamma_+^{-1}(z,v)=\frac{1}{1-zk}\psi(k),\\		&\Gamma_-(z,v)\psi^{\ast}(k)\Gamma_-^{-1}(z,v)=\left(1-\frac{z}{k}\right)\psi^{\ast}(k),\\ &\Gamma_+(z,v)\phi(k)\Gamma_+^{-1}(z,v)=\frac{1}{1-vk}\phi(k),\\
	&	\Gamma_-(z,v)\phi^{\ast}(k)\Gamma_-^{-1}(z,v)=\left(1-\frac{v}{k}\right) \phi^{*}(k).
	\end{align}
\end{lemma}
\begin{proof}
	From the Eq.(\ref{commutation relations 3.1}), one obtains
	\begin{align}\label{UC-TD}
		\Gamma_+(z,v)\psi(k)\Gamma_+^{-1}(z,v)
		&=e^{H_{+}(z,v)}	\psi(k) e^{-H_{+}(z,v)}
		\notag\\	
		&=\psi(k)+[H_{+}(z,v),\psi(k)]+\frac{1}{2!}[H_{+}(z,v),[H_{+}(z,v),\psi(k)]]+\ldots
		\notag\\	
		&=\psi(k)+\sum\limits_{n=1}^{\infty} \frac{(zk)^n}{n} \psi(k)+\frac{1}{2!}\xi_{\pm}^{2}(\mathbf{x}-\tilde{\partial}_\mathbf{y},k)\psi(k)+\ldots
		\notag\\
		&=\exp{\left(\sum\limits_{n=1}^{\infty} \frac{(zk)^n}{n}\right) }\psi(k).
	\end{align}
	By means of $\sum\limits_{n=1}^{\infty} \frac{(zk)^n}{n}=-\ln(1-zk)$, we obtain
	\begin{align}
		\Gamma_+(z,v)\psi(k)\Gamma_+^{-1}(z,v)=\frac{1}{1-zk}	\psi(k).
	\end{align}
	Using the similar procedure, we can prove other equations.
\end{proof}

\begin{lemma}\label{Gamma-state}
The actions of vertex operators $\Gamma_+(z,v)$  and $\Gamma_-(z,v)$ on state vectors satisfy
\begin{align}
	\Gamma_+(z,v)|\mu^1,\mu^2 \rangle &=
	\sum_{{\nu^1 \prec \mu^1}	\atop{\nu^2 \prec \mu^2}}
	z^{|\mu^1| - |\nu^1|}    v^{|\mu^2| - |\nu^2|}         |\nu^1,\nu^2 \rangle,\\
	\langle \mu^2,\mu^1|\Gamma_-(z,v) &=
	\sum_{{\nu^1 \prec \mu^1}	\atop{\nu^2 \prec \mu^2}}
	z^{|\mu^1| - |\nu^1|}    v^{|\mu^2| - |\nu^2|}     \langle \nu^2,\nu^1|.
\end{align}
\end{lemma}
\begin{proof}
By using the  Eq.$(\ref{right-state})$, we get
\begin{align}\label{Gamma_+-1}
	\Gamma_+(z,v)|\mu^1,\mu^2 \rangle
=\Gamma_+(z,v)\psi_{m_1^1} \ldots \psi_{m_{-l_1}^1}\phi_{m_1^2} \ldots \phi_{m_{-l_2}^2} | l_{1},l_{2} \rangle.
\end{align}
It follows from Lemma $\ref{Gamma-(k)}$ and the Eq.$(\ref{UC-feimisc})$ that
\begin{align}
\Gamma_+(z,v)\sum_{i\in\mathbb{Z}+1/2}\psi_i k^{-i-1/2}&=\sum_{n=0}^{\infty}  \sum_{j\in\mathbb{Z}+1/2}
(zk)^{n}  \psi_j k^{-j-1/2} \Gamma_+(z,v),\\
\Gamma_+(z,v)\sum_{i\in\mathbb{Z}+1/2}\phi_i k^{-i-1/2}&=\sum_{n=0}^{\infty}  \sum_{j\in\mathbb{Z}+1/2}
(vk)^{n}  \phi_j k^{-j-1/2} \Gamma_+(z,v).
\end{align}
Comparing the orders of $k$ on both sides, we have 
\begin{align}
\Gamma_+(z,v) \psi_i &=\sum_{n=0}^{\infty} z^{n} \psi_{(i+n)} \Gamma_+(z,v),\\
\Gamma_+(z,v) \phi_i &=\sum_{n=0}^{\infty} v^{n} \phi_{(i+n)} \Gamma_+(z,v).
\end{align}
Similarly,
\begin{align}
\Gamma_-(z,v) \sum_{n=0}^{\infty}  \psi_{(i-n)}^{*}  z^{n} &= \psi_i^{*}  \Gamma_-(z,v),\\
\Gamma_-(z,v) \sum_{n=0}^{\infty}  \phi_{(i-n)}^{*}  v^{n} &= \phi_i^{*}  \Gamma_-(z,v).
\end{align}
Combining $\Gamma_+(z,v) | l_{1},l_{2} \rangle=|  l_{1},l_{2} \rangle$, the Eq.$(\ref{Gamma_+-1})$ is rewritten as
\begin{align}
\Gamma_+(z,v)|\mu^1,\mu^2 \rangle
=&\left(\sum_{i_1=0}^{\infty} z^{i_1} \psi_{(i_1+m_1^1)}\right) \ldots
\left(\sum_{i_{-l_1}=0}^{\infty} z^{i_{-l_1}} \psi_{(i_{-l_1}+m_{-l_1}^1)}\right) \notag\\
&\left(\sum_{j_1=0}^{\infty} v^{j_1} \phi_{(j_1+m_1^2)}\right) \ldots
\left(\sum_{j_{-l_2}=0}^{\infty} v^{j_{-l_2}} \phi_{(j_{-l_2}+m_{-l_2}^2)}\right) | l_{1},l_{2} \rangle.
\end{align}
Since
\begin{align}
\left(\sum_{i_1=0}^{\infty} z^{i_1} \psi_{(i_1+m_1^1)}\right) &=
\left(\sum_{i_1=0}^{-m_1^{1}+m_2^{1}-1} z^{i_1}\psi_{(i_1+m_1^1)}\right)
+z^{(-m_1^{1}+m_2^{1})}\left(\sum_{l=0}^{\infty} z^{l} \psi_{(l+m_2^1)}\right),\\
\left(\sum_{j_1=0}^{\infty} v^{j_1} \phi_{(j_1+m_1^2)}\right) &=
\left(\sum_{j_1=0}^{-m_1^{2}+m_2^{2}-1} v^{j_1}\phi_{(j_1+m_1^2)}\right)
+v^{(-m_1^{2}+m_2^{2})}\left(\sum_{l=0}^{\infty} v^{l} \phi_{(l+m_2^2)}\right),
\end{align}
and based upon  the following identities
\begin{align}\label{KP-identity}
	\left(\sum_{i=0}^\infty z^i\psi_{(m+i)}\right)\left(\sum_{i=0}^\infty z^i \psi_{(m+i)}\right)=\left(\sum_{i=0}^\infty v^i\phi_{(m+i)}\right)\left(\sum_{i=0}^\infty v^i \phi_{(m+i)}\right)=0,
\end{align}
one obtains
\begin{align}
\Gamma_+(z,v)|\mu^1,\mu^2 \rangle=
&\left(\sum_{i_1=0}^{-m_1^{1}+m_2^{1}-1} z^{i_1}\psi_{(i_1+m_1^1)}\right) \ldots
\left(\sum_{i_{-l_1}=0}^{-m_{-l_1}^{1}+m_{(-l_1+1)}^{1}-1} z^{i_{-l_1}}\psi_{(i_{-l_1}+m_{-l_1}^1)}\right)\notag\\
&\left(\sum_{j_1=0}^{-m_1^{2}+m_2^{2}-1} v^{j_1}\phi_{(i_1+m_1^2)}\right)  \ldots
\left(\sum_{j_{-l_2}=0}^{-m_{-l_2}^{2}+m_{(-l_2+1)}^{2}-1} v^{j_{-l_2}}\phi_{(j_{-l_2}+m_{-l_2}^2)}\right)
| l_{1},l_{2} \rangle.
\end{align}
Note
\begin{align}
|\nu^1,\nu^2 \rangle&=\psi_{n_1^1} \ldots \psi_{n_{-l_1}^1}\phi_{n_1^2} \ldots \phi_{n_{-l_2}^2}| l_{1},l_{2} \rangle.
\end{align}
Consider the subscript relations,
\begin{align}
	m_{i}^{j} \leq n^i_{j} < m_{i+1}^{j}  \implies  \mu_i^{j}\geq \nu_i^{j} \geq \mu_{i+1}^{j}\quad
	\text{ for all } 1 \leq i \leq -l_{j}, \;  j=1,2,
\end{align}
then
\begin{align}
\Gamma_+(z,v)|\mu^1,\mu^2 \rangle=
&\left(\sum_{ m_1^{1}\leq n^1_{1} < m_2^{1}} z^{(n_1^{1}-m_1^{1})} \psi_{n^1_{1}}\right) \ldots
\left(\sum_{m_{-l_1}^{1} \leq    n_{-l_1}^{1}  < m_{(-l_1+1)}^{1}}
 z^{(n_{-l_1}^{1}-m_{-l_1}^{1})}   \psi_{ n_{-l_1}^{1} }\right)\notag\\
&\left(\sum_{m_1^{2}\leq n^2_{1} < m_2^{2}} v^{(n_1^{2}-m_1^{2})}\phi_{n^2_{1}}\right)  \ldots
\left(\sum_{   m_{-l_2}^{2} \leq  n_{-l_2}^{2}  < m_{(-l_2+1)}^{2}  }
    v^{(n_{-l_2}^{1}-m_{-l_2}^{2})} \phi_{ n_{-l_2}^{2} }\right)
| l_{1},l_{2} \rangle\notag\\
=&
\sum_{{\nu^1 \prec \mu^1}	\atop{\nu^2 \prec \mu^2}}
z^{|\mu^1| - |\nu^1|}    v^{|\mu^2| - |\nu^2|} |\nu^1,\nu^2 \rangle,
\end{align}
where $|\mu^j| - |\nu^j|=\sum\limits_{i=1}^{-l_{j}}(n^i_{j}- m_{i}^{j})$.

Analogous to the above derivation, we obtain
\begin{align}
	\langle \mu^2,\mu^1|\Gamma_-(z,v)
=\langle l_2,l_1|
&\left(  \sum_{m_{(-l_2+1)}^{2}   <   n_{-l_2}^{2}  \leq m_{-l_2}^{2}}
z^{(m_{-l_2}^{2}-n_{-l_2}^{2})} \phi_{ n_{-l_2}^{2} }^{*} \right)
\ldots
\left(\sum_{ m_2^{2}<n^1_{2} \leq m_1^{2}}
z^{(m_1^{2}-n_1^{2})}  \phi_{n^1_{2}}^{*} \right)\notag\\
&\left(  \sum_{m_{(-l_1+1)}^{1}   <   n_{-l_1}^{1}  \leq m_{-l_1}^{1}}
z^{(m_{-l_1}^{1}-n_{-l_1}^{1})} \psi_{ n_{-l_1}^{1} }^{*} \right)
 \ldots
\left(\sum_{ m_2^{1}<n^1_{1} \leq m_1^{1}}
z^{(m_1^{1}-n_1^{1})}  \psi_{n^1_{1}}^{*} \right)\notag\\
= \sum_{{\nu^1 \prec \mu^1}	\atop{\nu^2 \prec \mu^2}}
&z^{|\mu^1| - |\nu^1|}    v^{|\mu^2| - |\nu^2|}     \langle \nu^2,\nu^1|.
\end{align}
\end{proof}
\begin{remark}\label{Gamma+mu1,mu1-Gamma}
It is  showed that the action of vertex operators on state vectors produces interlacing $2-$partitions.
The case of $\mu^2=\emptyset$,  it can be deduced from the above derivation that
\begin{align}
	\Gamma_+(z,v)|\mu^1 \rangle &=
	\sum_{\nu^1 \prec \mu^1}
	z^{|\mu^1| - |\nu^1|}           |\nu^1 \rangle,\\
	\langle \mu^1|\Gamma_-(z,v) &=
	\sum_{\nu^1 \prec \mu^1}
	z^{|\mu^1| - |\nu^1|}    \langle \nu^1|.
\end{align}
The case of $\mu^1=\emptyset$ is similar.
\end{remark}

\subsection{The two-side generalized phase model on an infinite lattice}
Let us consider  the two-side generalized phase model by taking the number of lattice sites $M_{1},M_{2} \to \infty$.
\begin{lemma}\label{M-wq-n}
The basis vectors $|n^{1}\rangle^{(1)}\bigotimes|n^{2}\rangle^{(2)} \in \mathcal{V}$ and $	\prescript{(2)}{}{\langle n^{2}|}\bigotimes \prescript{(1)}{}{\langle n^{1}|} \in \mathcal{V}^{\ast} $ correspond to state vectors  $|\mu^1,\mu^2 \rangle \in \bar{\mathcal{F}}_{0,0}$ and $\langle \mu^2,\mu^1| \in \bar{\mathcal{F}}_{0,0}^*$ by maps $\mathcal{M}$ and $\mathcal{M}^*$,	respectively. The following equations hold
\begin{align}
&\mathcal{M} \left[ \lim_{M_{1},M_{2} \to \infty} \bar{\mathbb{B}}_{1}(z) \bar{\mathbb{B}}_{2}(v)|n^{1}\rangle^{(1)}\bigotimes|n^{2}\rangle^{(2)}\right] = \Gamma_-(z,v) | \nu^1,\nu^2 \rangle, \\
& \left[ \lim_{M_{1},M_{2}  \to \infty}
	\prescript{(2)}{}{\langle n^{2}|}\bigotimes \prescript{(1)}{}{\langle n^{1}|}  \bar{\mathbb{C}}_{2}(v)\bar{\mathbb{C}}_{1}(z) \right] \mathcal{M}^* =\langle \nu^2,\nu^1|  \Gamma_+(z,v).
\end{align}
\end{lemma}
\begin{proof}
Through the action of the mappings  $\mathcal{M}$ and $\mathcal{M}^*$, the limit $M_{1}, M_{2} \to \infty$ in Eqs.$(\ref{M-B1-B2-2})$ and $(\ref{M-C1-C2-2})$ is corresponds to limit $l_{1}, l_{2} \to \infty$ . Then we have
\begin{align}
\mathcal{M} \left[ \lim_{M_{1},M_{2} \to \infty} \mathbb{B}_{1}(z) \mathbb{B}_{2}(v)|n^{1}\rangle^{(1)}\bigotimes|n^{2}\rangle^{(2)} \right]
&=\sum_{{\nu^1 \prec \mu^1}	\atop{\nu^2 \prec \mu^2}}
z^{|\mu^1| - |\nu^1|}    v^{|\mu^2| - |\nu^2|} |\mu^1,\mu^2 \rangle,  \\
\left[ \lim_{M_{1},M_{2}  \to \infty}
	\prescript{(2)}{}{\langle n^{2}|}\bigotimes \prescript{(1)}{}{\langle n^{1}|}
 \mathbb{C}_{2}(v)\mathbb{C}_{1}(z) \right] \mathcal{M}^*
&=\sum_{{\nu^1 \prec \mu^1}	\atop{\nu^2 \prec \mu^2}}
z^{|\mu^1| - |\nu^1|}    v^{|\mu^2| - |\nu^2|}     \langle \mu^2,\mu^1|.
\end{align}
From Lemma $\ref{Gamma-state}$ and the Eq.(\ref{fock-scalar-product}), we have
\begin{align}
	\langle \mu^2,\mu^1|\Gamma_-(z,v) |\nu^1,\nu^2 \rangle
	&=
	z^{|\mu^1| - |\nu^1|}    v^{|\mu^2| - |\nu^2|}, \\
	\langle \nu^2,\nu^1|\Gamma_+(z,v)|\mu^1,\mu^2 \rangle
	&=
	z^{|\mu^1| - |\nu^1|}    v^{|\mu^2| - |\nu^2|}  .
\end{align}
Thus
\begin{align}
\Gamma_-(z,v) |\nu^1,\nu^2 \rangle
&=\sum_{{\nu^1 \prec \mu^1}	\atop{\nu^2 \prec \mu^2}}
z^{|\mu^1| - |\nu^1|}    v^{|\mu^2| - |\nu^2|} |\mu^1,\mu^2 \rangle, \label{Gamma-right} \\
\langle \nu^2,\nu^1|\Gamma_+(z,v)
&=\sum_{{\nu^1 \prec \mu^1}	\atop{\nu^2 \prec \mu^2}}
z^{|\mu^1| - |\nu^1|}    v^{|\mu^2| - |\nu^2|}     \langle \mu^2,\mu^1|.\label{Gamma+left}
\end{align}
\end{proof}
An analogous discussion of the Remark $\ref{Gamma+mu1,mu1-Gamma}$ yields
\begin{align}
	\Gamma_-(z,v) |\nu^1 \rangle
	&=\sum_{\nu^1 \prec \mu^1}
	z^{|\mu^1| - |\nu^1|}     |\mu^1 \rangle, \label{Gamma-right-1}\\
	\langle \nu^1|\Gamma_+(z,v)
	&=\sum_{\nu^1 \prec \mu^1}
	z^{|\mu^1| - |\nu^1|}      \langle \mu^1|\label{Gamma+left-1},
\end{align}
and the case of $\nu^1=\emptyset$ is similar.

Considering the scalar product  with $M_{1},M_{2} \to \infty$, we obtain
\begin{align}
&\lim\limits_{M_{1},M_{2}\to \infty}	
S(\{\mathbf{x^1}\}, \{\mathbf{y^2}\}, \{\mathbf{z^1}\}, \{\mathbf{v^2}\})\notag\\
	=&\lim\limits_{M_{1},M_{2}\to \infty}
	\prescript{(2)}{}{\langle 0|}\bigotimes \prescript{(1)}{}{\langle 0|}
	\rprod_{j=1}^{N_2} \bar{\mathbb{C}}_{2}(y_j)
	\rprod_{i=1}^{N_1} \bar{\mathbb{C}}_{1}(x_i)
	\lprod_{l=1}^{N_1} \bar{\mathbb{B}}_{1}(z_i)
	\lprod_{k=1}^{N_2} \bar{\mathbb{B}}_{2}(v_j)
	|0\rangle^{(1)}\bigotimes|0\rangle^{(2)}\notag\\
	=&
	\sum_{{\pi^{1} \subseteq [N_{1}, N_{1},\infty]}
		\atop{\pi^{2} \subseteq [N_{2}, N_{2}, \infty]}}
	A_{\pi^{1} }(\{\mathbf{x^1}\}, \{\mathbf{z^1}\}) A_{\pi^{2} }(\{\mathbf{y^2}\}, \{\mathbf{v^2}\}).
\end{align}
Setting $N=\max\{N_{1},N_{2}\}$, by using the Eq.$(\ref{Gamme-commutation-relations})$ and the Lemma $\ref{M-wq-n}$, one gets
\begin{align}\label{scalar-product-Gamma1-Gamma2-m-wq}
	&\lim\limits_{M_{1},M_{2}\to \infty}
	\prescript{(2)}{}{\langle 0|}\bigotimes \prescript{(1)}{}{\langle 0|}
	\rprod_{j=1}^{N_2} \bar{\mathbb{C}}_{2}(y_j)
	\rprod_{i=1}^{N_1} \bar{\mathbb{C}}_{1}(x_i)
	\lprod_{l=1}^{N_1} \bar{\mathbb{B}}_{1}(z_i)
	\lprod_{k=1}^{N_2} \bar{\mathbb{B}}_{2}(v_j)
	|0\rangle^{(1)}\bigotimes|0\rangle^{(2)}\notag\\
	=&
\langle 0|
\Gamma_+(x_{N},y_{N})\cdots\Gamma_+(x_{1},y_{1})
\Gamma_-(z_1,v_1)\cdots\Gamma_-(z_{N},v_{N})
	|0\rangle\notag\\
	=&	\prod_{i,l=1}^{N_1} \frac{1}{1 - z_{l} x_{j}} \prod_{j,k=1}^{N_2} \frac{1}{1 - y_{j} v_{k}},
\end{align}
which means
\begin{align}\label{scalar-product-M}
	\sum_{{\pi^{1} \subseteq [N_{1}, N_{1},\infty]}
		\atop{\pi^{2} \subseteq [N_{2}, N_{2}, \infty]}}
		A_{\pi^{1} }(\{\mathbf{x^1}\}, \{\mathbf{z^1}\}) A_{\pi^{2} }(\{\mathbf{y^2}\}, \{\mathbf{v^2}\})
	=\prod_{i,l=1}^{N_1} \frac{1}{1 - z_{l} x_{j}} \prod_{j,k=1}^{N_2} \frac{1}{1 - y_{j} v_{k}}.
\end{align}
Moreover, it follows from the Eq.$(\ref{scalar-product-S1-S2})$ that
\begin{align}
\sum_{{ \mu^{1} \subseteq [N_1,\infty]}\atop{\mu^{2} \subseteq [N_2,\infty]}}
s_{\mu^{1}}\{\mathbf{x^1}\}  s_{\mu^{1}}\{\mathbf{z^1}\}
s_{\mu^{2}}\{\mathbf{y^2}\}  s_{\mu^{2}}\{\mathbf{v^2}\}	
	=\prod_{i,l=1}^{N_1} \frac{1}{1 - x_{i}z_{l} } \prod_{j,k=1}^{N_2} \frac{1}{1 - y_{j} v_{k}}.
\end{align}
We will show that the Eq.$(\ref{scalar-product-M})$ can be rewritten as the generating function of UC plane partitions with the double scaling limit.
Set
\begin{align}\label{set-z}
	x_{n}=z_{n}=p^{n-\frac{1}{2}}, ~ y_{m}=v_{m}=q^{m-\frac{1}{2}},
	\quad \text{for all } 1 \leq n \leq N_1, ~ 1 \leq m \leq N_2.
\end{align}
Substituting the Eq.$(\ref{set-z})$ into the Eq.$(\ref{scalar-product-Gamma1-Gamma2-m-wq})$, we have
\begin{align}
	&\langle 0|
	\Gamma_+(p^{N- \frac{1}{2}},q^{N - \frac{1}{2}})\cdots\Gamma_+(p^{\frac{1}{2}},q^{\frac{1}{2}})
	\Gamma_-(p^{\frac{1}{2}},q^{\frac{1}{2}})\cdots\Gamma_-(p^{N - \frac{1}{2}},q^{N - \frac{1}{2}})
	|0\rangle \notag\\
	=&\prod_{i,l=1}^{N_1} \frac{1}{1 - p^{l+i-1}} \prod_{j,k=1}^{N_2} \frac{1}{1 - q^{j+k-1}}.
\end{align}
From Eqs.$(\ref{Gamma-right})$ and $(\ref{Gamma+left})$, we assume
\begin{align}
	\Gamma_-(p^{i-\frac{1}{2}},q^{i-\frac{1}{2}}) |\nu^1,\nu^2 \rangle
	&=\sum_{{\nu^1 \prec \pi^1_{i-1}}	\atop{\nu^2 \prec \pi^2_{i-1}}}
	p^{(i-\frac{1}{2})(|\pi^1_{i-1}| - |\nu^1|)}  q^{(i-\frac{1}{2})(|\pi^2_{i-1}| - |\nu^2|)} |\pi^1_{i-1},\pi^2_{i-1} \rangle,  \\
	\langle \nu^2,\nu^1|\Gamma_+(p^{i-\frac{1}{2}},q^{i-\frac{1}{2}})
	&=\sum_{{\nu^1 \prec \pi^1_{-i+1}}	\atop{\nu^2 \prec \pi^2_{-i+1}}}
	p^{(i-\frac{1}{2})(|\pi^1_{-i+1}| - |\nu^1|)}    q^{(i-\frac{1}{2})(|\pi^2_{-i+1}| - |\nu^2|)}     \langle \pi^2_{-i+1},\pi^1_{-i+1}|.
\end{align}
By means of Eqs.$(\ref{Gamma-right})-(\ref{Gamma+left-1})$, one gets
\begin{align}
&\Gamma_-(p^{\frac{1}{2}},q^{\frac{1}{2}})\cdots\Gamma_-(p^{N - \frac{1}{2}},q^{N - \frac{1}{2}})
|0\rangle \notag\\
=&	\sum_{{ [N_{1}, \infty] \ni \pi^{1}_0 \succ \cdots \succ \pi^{1}_{N_{1}} = \emptyset}
	\atop{[N_{2}, \infty] \ni \pi^{2}_0 \succ \cdots \succ \pi^{2}_{N_{2}} = \emptyset}}
\prod_{j=1}^{N_1}	p^{(j-\frac{1}{2})(|\pi^1_{j-1}| - |\pi^1_{j}|)}
\prod_{k=1}^{N_2}	q^{(l-\frac{1}{2})(|\pi^1_{k-1}| - |\pi^2_{k}|)}  |\pi^{1}_0,\pi^{2}_0\rangle, \\
&\langle 0|
\Gamma_+(p^{N - \frac{1}{2}},q^{N - \frac{1}{2}})\cdots
\Gamma_+(p^{\frac{1}{2}},q^{\frac{1}{2}})\notag\\
=&\sum_{{ \emptyset=\pi^{1}_{-N_{1}} \prec \cdots \prec \pi^{1}_0 \in [N_{1}, \infty]  }
	\atop{\emptyset=\pi^{2}_{-N_{2}} \prec \cdots \prec \pi^{2}_0  \in [N_{2}, \infty]  }}
\prod_{i=1}^{N_1} p^{(i-\frac{1}{2})(|\pi^1_{-i+1}| - |\pi^1_{i}|)}
\prod_{l=1}^{N_2}	 q^{(l-\frac{1}{2})(|\pi^2_{-l+1}| - |\pi^2_{l}|)} \langle\pi^{2}_0, \pi^{1}_0 |,
\end{align}
then
\begin{align}
	&\langle 0|
	\Gamma_+(p^{N_1 - \frac{1}{2}},q^{N_2 - \frac{1}{2}})\cdots\Gamma_+(p^{\frac{1}{2}},q^{\frac{1}{2}})
	\Gamma_-(p^{\frac{1}{2}},q^{\frac{1}{2}})\cdots\Gamma_-(p^{N_1 - \frac{1}{2}},q^{N_2 - \frac{1}{2}})
	|0\rangle \notag\\
	=&\sum_{{\pi^{1} \subseteq [N_{1}, N_{1},\infty]}	\atop{\pi^{2} \subseteq [N_{2}, N_{2}, \infty]}}		p^{|\pi^{1}|}q^{|\pi^{2}|}.
\end{align}
It is easy to show that
\begin{align}
\sum_{{\pi^{1} \subseteq [N_{1}, N_{1},\infty]}	\atop{\pi^{2} \subseteq [N_{2}, N_{2}, \infty]}}	
	p^{|\pi^{1}|}q^{|\pi^{2}|}	
=\prod_{i,l=1}^{N_1} \frac{1}{1 - p^{l+i-1}} \prod_{j,k=1}^{N_2} \frac{1}{1 - q^{j+k-1}}.
\end{align}
Considering the  number of particles $N_{1},N_{2} \to \infty$,
\begin{align}
	\sum_{\begin{smallmatrix}\pi^{1}\text{ and }\pi^{2}\text{ are}\\\text{plane partitions}\end{smallmatrix}}
	p^{|\pi^{1}|}q^{|\pi^{2}|}
&=\lim\limits_{N_{1},N_{2}\to \infty}	
\left(\prod_{i=1}^{N_1} \frac{1}{1 - p^{i}}  \cdots  \prod_{i=N_1}^{2N_1} \frac{1}{1 - p^{i}}\right)
\left(\prod_{j=1}^{N_2} \frac{1}{1 - q^{j}}  \cdots  \prod_{j=N_2}^{2N_2} \frac{1}{1 - q^{j}}\right)\notag\\
	&=\prod_{n=1}^\infty\left(\frac1{1-p^n}\right)^n\prod_{m=1}^\infty\left(\frac1{1-q^m}\right)^m.
\end{align}
Therefore, the  generating function of UC plane partitions is obtained with the double scaling limit.
\section{Conclusions and discussions}
This paper is concerned with the  relation between the two-side generalized phase model and boxed UC plane partitions.
We define boxed UC plane partitions and study the representation of two-side generalized phase algebras.
Meanwhile, with the help of maps $\mathcal{M}$ and $\mathcal{M}^*$,  interlacing boxed $2-$partitions have been generated.
By means of the  scalar product of the two-site generalized phase model, we discuss the generating function of  boxed UC plane partitions and its double scaling limit case $(M_{j},N_{j} \to \infty, j=1,2)$. It should be pointed out that
relations between the limiting forms for the scalar product of the five-vertex model and   complete symmetric polynomials have been presented \cite{five}.
The limiting form of the scalar product and its determinantal representation in the two-site generalized phase model deserves further exploration. 
\section{Acknowledgements}
This article is dedicated to Professor Ke Wu in Capital Normal University in celebration of his 80th birthday. And this work is supported by the National Natural Science Foundation of China (Grant No.12061051) and the Program for Young Talents of Science and Technology in Universities of Inner Mongolia Autonomous Region (Grant No. NJYT23096). The authors gratefully acknowledge the support of Professor Ke Wu and Professor Weizhong Zhao at Capital Normal University, China.

\end{document}